\newcount\draft\draft=0		
\newcount\cameraready\cameraready=0	

\ifnum\cameraready=0
 \documentclass[10pt]{sig-alternate-10pt}
\else
 \documentclass{sig-alternate}
\fi
\usepackage{graphicx,epsfig,sped,times,wide}
\usepackage{algorithm,url,array}
\usepackage{clrscode}
\usepackage[noeka]{mathrmletter}
\usepackage{subfigure}
\usepackage{multirow}
\usepackage{multirow}
\usepackage{appendix}
\usepackage[noend]{algorithmic}
\usepackage{verbatim}
\usepackage[numbers,sort]{natbib}

\ifnum\draft=0
\fi

\newcommand{\name} {Rapid-Link} 

\newenvironment{Itemize}%
{\begin{itemize}%
\setlength{\itemsep}{0pt}%
\setlength{\topsep}{0pt}%
\setlength{\partopsep}{0pt}%
\setlength{\parskip}{0pt}}%
{\end{itemize}}
{\begin{itemize}%
\setlength{\itemsep}{0pt}%
\setlength{\topsep}{0pt}%
\setlength{\partopsep}{0pt}%
\setlength{\parskip}{0pt}}%
{\end{itemize}}
\makeatletter
  \newcommand\figcaption{\def\@captype{figure}\caption}
  \newcommand\tabcaption{\def\@captype{table}\caption}
\makeatother

\newtheorem{theorem}{Theorem}[section]
\newtheorem{lemma}[theorem]{Lemma}
\newtheorem{proposition}[theorem]{Proposition}

\newenvironment{remark}[1][Remark]{\begin{trivlist}
\item[\hskip \labelsep {\bfseries #1}]}{\end{trivlist}}

\newtheorem{claim}[theorem]{Claim}

\newcommand{\wh}{\widehat}

\newcommand{\bF}{{\bf F}}
\newcommand{\bA}{{\bf A}}
\newcommand{\bB}{{\bf B}}
\newcommand{\bP}{{\bf P}}

\newcommand{\xqed}{\nobreak \ifvmode \relax \else
      \ifdim\lastskip<1.5em \hskip-\lastskip
      \hskip1.5em plus0em minus0.5em \fi \nobreak
      \vrule height0.75em width0.5em depth0.25em\fi}

\newcommand{\xref}[1]{\S\ref{#1}}

\newcommand{\textred}[1]{\textcolor{red}{#1}}
\ifx\noeditingmarks\undefined
   \newcommand{\pgwrapper}[2]{\textred{#1: #2}}
\else
   \newcommand{\pgwrapper}[2]{}
\fi

\usepackage{xxx} 

\makeatletter

\global\def\section{\@startsection {section}{1}{\z@}%
                                   {2ex \@plus 1ex \@minus .1ex}%
                                   {1ex \@plus.2ex}%
                                   {\normalfont\bfseries\scshape\fontsize{11}{13}\selectfont}}
\global\def\subsection{\@startsection{subsection}{2}{\z@}%
                                     {2ex\@plus 1ex \@minus .1ex}%
                                     {1ex \@plus .2ex}%
                                     {\normalfont\bfseries\fontsize{10}{12}\selectfont}}
\global\def\subsubsection{\@startsection{subsubsection}{3}{\z@}%
                                     {2ex\@plus 1ex \@minus .1ex}%
                                     {1ex \@plus .2ex}%
                                     {\normalfont\itshape\fontsize{10}{12}\selectfont}}

\global\def\@maketitle{%
  \newpage
  \begin{center}%
  \let \footnote \thanks
  \null
    \vskip -.3em%
    {\LARGE\bf \@title \par}%
    \vskip 1em%
    {\large
      \lineskip .5em%
      \begin{tabular}[t]{c}%
        \@author
      \end{tabular}\par}%
    \vskip 1em%
    {\large \@date}%
  \end{center}%
  \par
  \vskip 2em}

\begin{document}

\title{
Agile Millimeter Wave Networks with Provable Guarantees 
}

\newcommand{\supsym}[1]{\raisebox{6pt}{{\footnotesize #1}}}

\numberofauthors{1} 
\author{
\alignauthor  Haitham Hassanieh$^{\dagger 1}$ \quad  Omid Abari$^{\dagger 2}$ \quad Michael Rodreguez$^2$ \quad Mohammed Abdelghany$^3$ \quad Dina Katabi$^2$ \quad Piotr Indyk$^2$\\
\affaddr{$^1$ UIUC \quad $^2$ MIT \quad $^3$ UCSB} \\ 
\affaddr{$^\dagger$Co-primary Authors}
}


\date{}
\maketitle

\newcommand*\wrapletters[1]{\wr@pletters#1\@nil}
\def\wr@pletters#1#2\@nil{#1\allowbreak\if&#2&\else\wr@pletters#2\@nil\fi}
\newcommand\ie{\textit{i.e.}}
\newcommand\newchanges[1]{{#1}}

\begin{sloppypar}

\noindent{\bf Abstract}-- 
There is much interest in integrating millimeter wave radios (mmWave) into wireless LANs and 5G cellular networks to benefit from their multiple GHz of available spectrum. Yet unlike existing technologies, e.g., WiFi, mmWave radios require highly directional antennas. 
Since the antennas have pencil-beams, the transmitter and receiver need to align their antenna beams before they can communicate.  Existing solutions scan the entire space to find the best alignment. Such a process has been shown to introduce up to seconds of delay, and is unsuitable for wireless networks where an access point has to quickly switch between users and accommodate mobile clients.

This paper  presents \name, a new protocol that can find the best mmWave beam alignment without scanning the space.  Given all possible directions for setting the antenna beam, \name\  provably finds the optimal direction in logarithmic number of measurements. Further, \name\ works within the existing 802.11ad standard for mmWave LAN, and can support both clients and access points. We have implemented \name\ in a mmWave radio and evaluated it empirically. Our results show that it reduces beam alignment delay by orders of magnitude. In particular, for highly directional mmWave devices operating under 802.11ad, the delay drops from over a second to 2.5~ms.


\section{Introduction}\label{sec:intro}
The ever-increasing demands for mobile and wireless data have placed a huge
strain on wireless networks~\cite{cisco, ericsson, umts}.  Millimeter wave
(mmWave) frequency bands address this problem by offering multi-GHz of
unlicensed bandwidth,  200$\times$ more than the bandwidth allocated to today's
WiFi and cellular networks~\cite{pi2011introduction, rangan2014millimeter}.
They range from the 24GHz ISM band all the way to hundreds of
GHz~\cite{infineon_24ghz,FCC_24ghz}.  They are expected to play a central role
in dealing with increased multimedia traffic, the introduction of new high
data-rate applications such as virtual reality, and the anticipated surge in
IoT wireless devices~\cite{5Gvision,MoVR_hotnets, siBeam_VR}. This role has
been cemented with new standards that incorporate mmWave technologies into 5G
cellular networks~\cite{rangan2014millimeter, mmW5G, mmw_AD}, and 802.11
wireless LAN~\cite{802.11ad}.

Millimeter wave radios however do not play well with mobile devices or dynamic
environment, a key challenge that has been emphasized in the
standards~\cite{802.11ad_paper,802.11ad}. Specifically, mmWave signals
attenuate quickly with distance; hence they need to use highly directional
antennas to focus their power on the receiver.  Luckily, due to their small
wavelength (millimeter scale),  it is possible to pack hundreds or thousands of
antennas in a small area, creating an array with many antennas, and hence a
very narrow beam, as shown in Figs.~\ref{fig:intro}(a) and (b).   Yet, since
the beam is very narrow, communication is possible only when the transmitter's
and receiver's beams are well aligned. Current solutions for aligning the beams
scan the entire space, trying various beam alignments until they find the best
one.  This process can take up to several seconds~\cite{60GHzMobicom,WiMi}. 
Such a long delay makes the deployment of mmWave links infeasible in wireless
networks, where the access point has to keep realigning its beam to switch
between users and accommodate mobile clients.


To understand the problem, consider a phased array with many antennas. Its
beam-width can be a few degrees or even smaller.  The naive approach to finding
the best alignment would have the transmitter and receiver scan the 3D space
with their beams to find the direction of maximum power, as shown in
Fig.~\ref{fig:intro}(b). The receiver has to repeat the scan for each choice of
beam direction on the transmitter side. Thus, the complexity of this exhaustive
search is $O(N^2)$, where $N$ is the number of possible beam directions.  To
speed up the search, the 802.11ad standard decouples the steering at the
transmitter and receiver.  In particular, the transmitter starts with a
quasi-omnidirectional beam, while the receiver scans the space for the best
beam direction. The process is then reversed to have the transmitter scan the
space while keeping the receiver quasi-omnidirectional~\cite{802.11ad,
mmw_heir6} (see~\xref{sec:baselines} for the details).   This approach reduces
the search complexity to $O(N)$. Still, for a beam of a few degrees, the delay
can be hundreds of milliseconds to seconds~\cite{60GHzMobicom,WiMi}, which
would easily stall realtime applications.

But, can one identify the best alignment without scanning the space of all
possible signal directions? In principle, "yes".  There is much past work that
shows that mmWave signals travel along a small number of paths (e.g., the
direct path from transmitter to receiver and a few
reflections)~\cite{rangan2014millimeter,mmw_measurement2}. This means that the
space of possible signal directions is sparse. One would hope to use the sparse
recovery theory to estimate the direction of the best alignment using a
logarithmic number of measurements~\cite{CRT,Don,SFFT1}, hence avoiding
excessive delays.  

\begin{figure*}
\centering
\begin{tabular}{ccccc}
	\includegraphics[angle=0,width=1.4in]{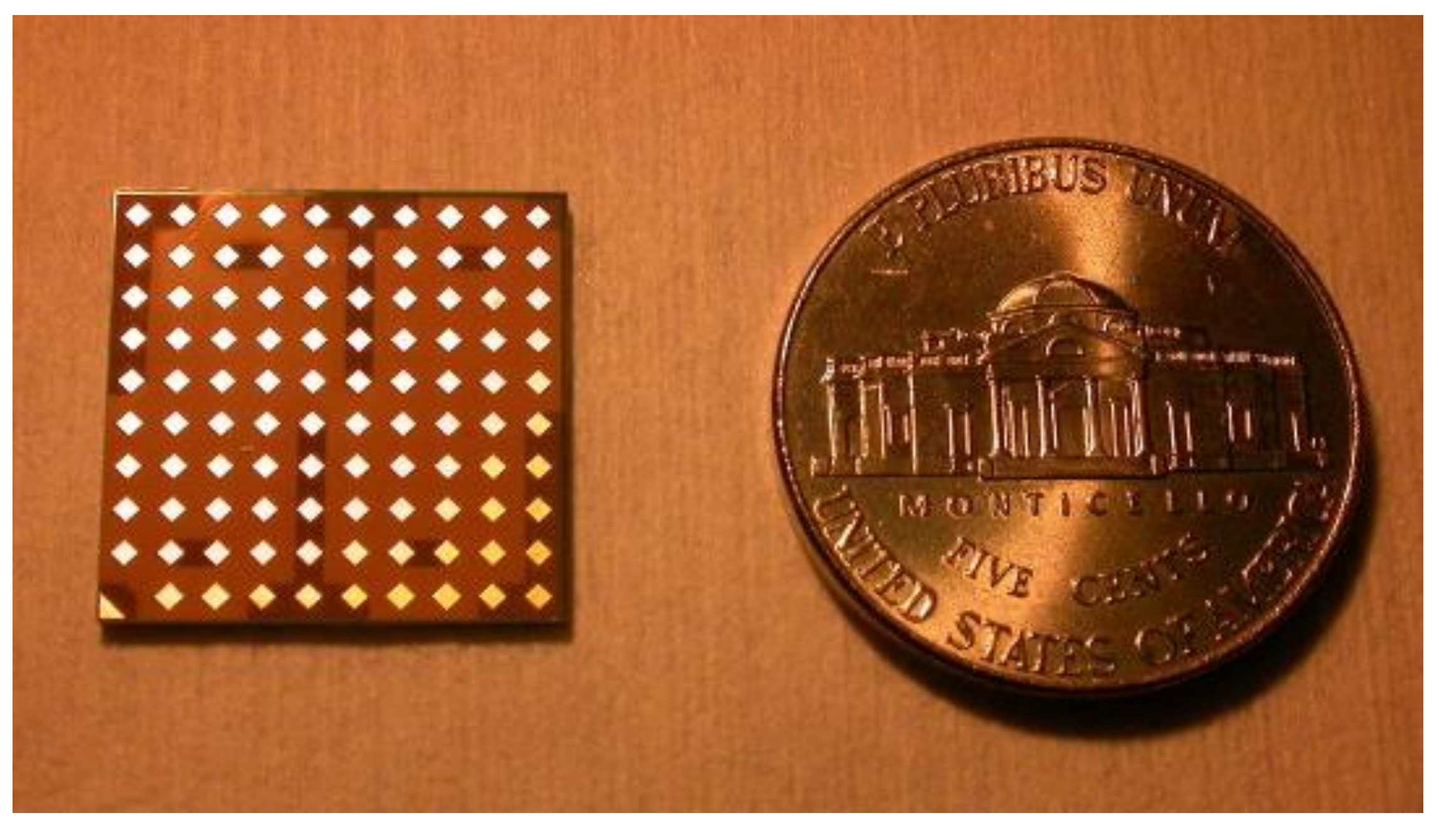} && \quad 
	\includegraphics[angle=0,width=1.8in]{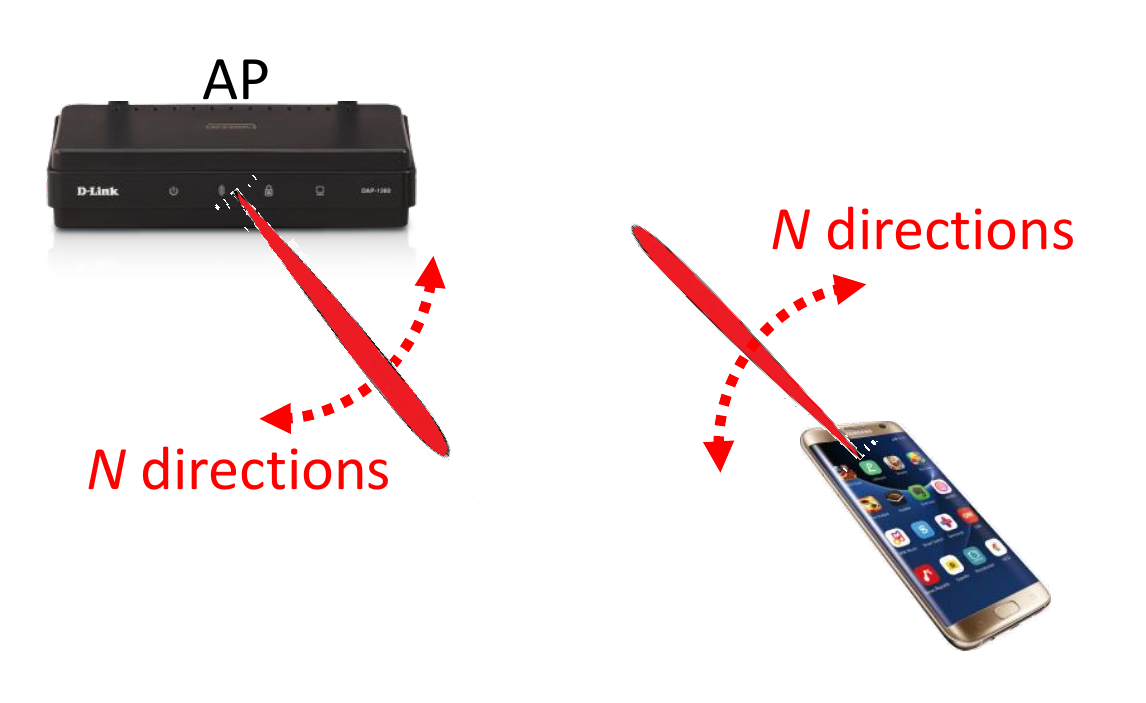} && \quad 
	\includegraphics[angle=0,width=1.6in]{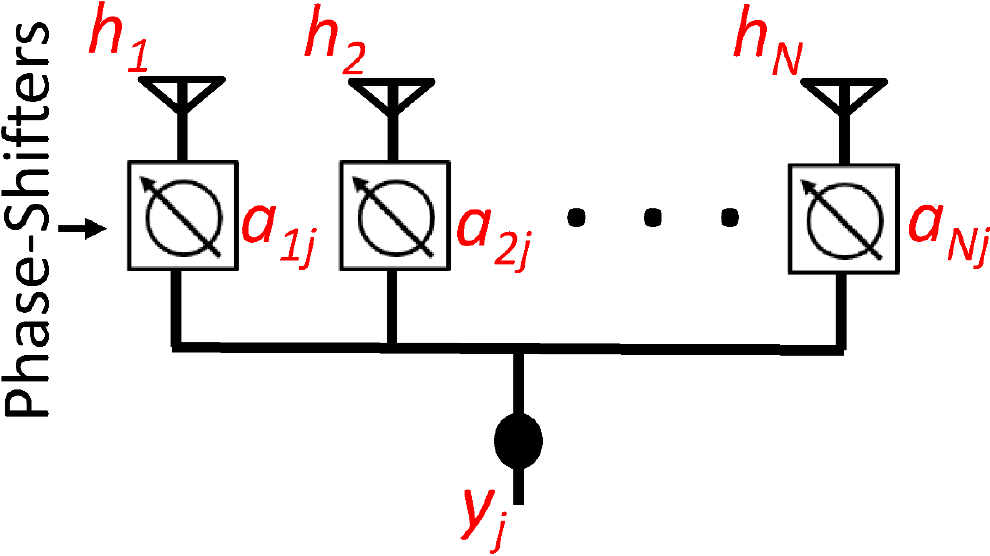} \\
(a) A mmWave phased array&&
(b) Beam alignment in mmWave radios&&
(c) Phased array's architecture.\\
\end{tabular}
\vskip -0.1in
\caption{\textbf{mmWave Communication} (a) An example of mmWave phased arrays where hundreds of antennas are packed in a small area~\cite{IBM_array}. (b) mmWave radios need to find the best alignment between the transmitter's and receiver's beams to establish a communication link. (c) mmWave phased arrays use a set of phase shifters to steer the beam electronically.}
\vskip -0.1in
\label{fig:intro}
\end{figure*}

\vskip 0.06in\noindent
{\bf Problem Formalization:}
To be more concrete, let us formalize the problem. The objective of beam alignment is to measure the signal power along each spatial direction. Let $\bf{x}$ be an $N$-element vector that denotes the signal along various spatial directions. Since in practice the signal arrives only along few directions $K$, we can say that $\bf{x}$ is {\it K-sparse}.  Let, $h_i$ be the signal at the i$^{th}$ antenna, as shown in Fig.~\ref{fig:intro}(c). Based on the standard antenna array equation~\cite{ArrayTrack}, we can write ${\bf h} = \bF' {\bf x}$, where $\bF'$ is the inverse Fourier transform matrix.
We can steer the antenna beam by applying a phase shift to each antenna $a_i = e^{-j 2\pi \phi_i}$ (see Fig.~\ref{fig:intro}(c)).  For each setting of the phase shifters, we can measure the received power as $y_j = |{\bf a}_j \bF' {\bf x}|$, where the notation $|.|$ refers to the magnitude of the signal (i.e., the power),\footnote{The power is the magnitude square; hence knowing the magnitude is the same as knowing the power.}
and ${\bf a}_j$ is a vector whose elements are the applied phase shifts. In 802.11ad, each measurement corresponds to sending a special frame. Each time a frame is sent, the signal incurs a random phase due to the Carrier Frequency Offset (CFO)
between the transmitter and receiver~\cite{802.11ad,MegaMIMO} (see~\xref{sec:problem} for details). Thus, one cannot compare the phase of two measurements; only the magnitude of the measurements is relevant. Since we want to know the power along each spatial direction, the problem can be formulated as:
$$
\text{estimate each $|x_i|$, given measurements~~} y_j = |{\bf  a}_j \bF' {\bf x}|.
$$ 
Of course, one way to solve this problem is to use $N$ measurements, each time setting ${\bf a}$ to one row of the Fourier matrix. This corresponds to measuring one direction every time, as proposed in 802.11ad.  Alternatively, one could leverage that the vector $\bf{x}$ is sparse, and hope to solve the problem in a logarithmic number of measurements. Unfortunately, however, using off-the-shelf algorithms like compressive sensing or the sparse FFT does not work since neither of them deal with the scenario where the measurements return the magnitude of the complex signal, i.e., the presence of the $|.|$ term~\cite{CRT,Don,gilbert2014recent}.
To the best of our knowledge, no algorithm with provable logarithmic guarantees exist for this problem.

\vskip 0.06in\noindent
{\bf Our Design \& Analysis:}
This paper introduces \name, a novel protocol for fast beam alignment in mmWave networks. \name\  provably finds the best alignment in a logarithmic number of measurements. At a high-level, it works as follows:  Instead of creating a narrow beam and sampling the power along one spatial direction each time,  \name\ manipulates the phase shifters to create multi-armed beams, which can sample multiple spatial directions simultaneously (see Fig~\ref{fig:example}(a)). When sampling multiple directions simultaneously, one cannot tell which direction has produced the resulting power. \name\ however uses a carefully-designed combination of such multi-armed beams, 
which quickly provide enough information to identify the signal power along all directions, as we detail in~\xref{sec:alg}. We formally analyze \name\ and prove that it can deliver the best alignment in $O(K \log N)$ measurements, where $K$ is the number of paths traveled by the signal.  Since $K$ is typically 2 or 3 paths~\cite{rangan2014millimeter,mmw_measurement2},  \name\ can significantly reduce the beam alignment delay. 

 \name\ has additional important features:
\begin{Itemize}
\item \name\ is compatible with the 802.11ad protocol, i.e., a \name\ device can work with a non-\name\ device to find the best alignment while using the 802.11ad protocol.  In this case, the \name\ device finds the best alignment on its side in a logarithmic number of measurements whereas the traditional 802.11ad device takes a linear number of measurements. 
\item When both transmitter and receiver are equipped with \name, they can coordinate the search and find the best transmit and receive beam directions in $O(K^2 \log N)$ (as opposed to $O(K^2 \log^2N)$). They can do so without using quasi-omnidirectional beams, which are known to reduce the quality of the alignment~\cite{BoonBane,hosoya2015multiple}.
\item Finally, by the results of \cite{ba2010lower,price2011}, our measurement complexity of $O(K\log N)$ is asymptotically {\em optimal} for small $K$ --i.e., it cannot be further reduced.
\end{Itemize}

\vskip 0.06in\noindent
{\bf Implementation \& Empirical Results:}
We have evaluated \name\ empirically using mmWave radios, each equipped with a phased array that has 8 antennas.  We have implemented the design as a daughterboard for the USRP software radio, which enables easy manipulation of mmWave signals using standard GNU-radio software.  We also use simulations to explore its scaling behavior to large arrays with hundreds of antennas, which are expected in the future~\cite{qualcomm}. We compare \name\ with two baselines: an exhaustive scan of the space to find the best beams, and the quasi-omnidirectional search proposed in the 802.11ad standard. Our evaluation reveals the following
findings. 
\begin{Itemize}
\item In comparison with the exhaustive search, \name\ reduces the search
time by one to three orders of magnitude for array sizes that range from 8
antennas to 256 antennas. In comparison to the quasi-omnidirectional search,
\name\ reduces the number of measurements by 1.5$\times$ to 16.4$\times$ for the same 
range of array sizes. In particular, for large arrays with 256 antennas, \name\ reduces the alignment delay from over a second to 2.5ms. 
\item  The quasi-omnidirectional search yields poor performance in scenarios with multi-path effects. This is because using a quasi-omnidirectional antenna allows the signals along different paths to combine destructively, which yields low power and prevents accurate detection of the best signal direction. Further, due to imperfections in the quasi-omnidirectional patterns~\cite{BoonBane,hosoya2015multiple}, some paths can get attenuated and hence this approach can choose the wrong direction to align its beam.  In contrast, \name\ performs well both in single path and multipath scenarios. 
\end{Itemize}

\vskip 0.06in\noindent
{\bf Contributions:}
This paper provides the first mmWave beam alignment algorithm with provably logarithmic measurements for  the phased-array architecture commonly used in mmWave access points and clients. 
The paper also demonstrates through an implementation and empirical evaluation the feasibility of the design and its practical gains.

\section{Related Work}
\label{sec:related}

\vskip 0.06in \noindent {\bf (a) mmWave Solutions \& Research:}
Research on fast beam alignment for mmWave can be divided into two classes: empirical and simulation-based. Past empirical work has demonstrated the large delays incurred during beam alignment~\cite{60GHzMobicom,WiMi}. It also proposed failover protocols that switch to the next best beam when the current beam becomes blocked~\cite{BeamSpy,MOCA}.  This approach however assumes that the signal propagation paths are known a priori and hence one can quickly switch to a failover direction.  

Much of the previous work on fast beam alignment is simulation-based.  Most of this work proposes enhancements to the standard that impose a form of hierarchy to speed up the search~\cite{mmw_heir1, mmw_heir2, mmw_heir3, mmw_heir4, mmw_heir5, mmw_heir6, mmw_heir7}. However, in practice, hierarchical search requires feedback from the receiver to guide the transmitter  at every stage of the hierarchy, which incurs significant protocol delay. 

Our work is closest to past work that leverage compressive sensing to speed up the search for best beam alignment. 
The work in~\cite{mmw_hybrid1, mmw_hybrid2} requires a more complex architecture with a quadratic number of phase shifters, and multiple transmit receive chains (typically 10 to 15~\cite{mmw_hybrid1}).  Despite being less constrained, the best known results for such complex architecture can guarantee a logarithmic number of measurements {\it only} for scenarios with no multi-path ( $K$ is strictly $1$), and  only for the average case error~\cite{mmw_hybrid1}. In comparison, \name\ can provably find the best alignment in a logarithmic number of measurements even when the signal experiences multipath, and its guarantees apply to the worst case behavior. Furthermore, since \name\ applies to the more restricted architecture which has a linear number of phase shifters and only one transmit-receive chain, its guarantees naturally extend to this more complex architecture. 

Some past theoretical work applies the standard compressive sensing and assumes it can correctly obtain the phase of the measurements~\cite{mmw_cs1,mmw_cs2}.  This approach does not work with practical 802.11ad or cellular radios because it ignores CFO (Carrier Frequency Offset) which corrupts the phase of the measurements~\cite{MegaMIMO}. Note that correcting for CFO across measurement frames is neither possible in the current 802.11ad standard nor easy. For example, a small offset of 10 parts per million at such high frequencies can cause drastic phase misalignment in less than hundred nanoseconds.

Finally, some companies such as TP-Link and SiBeam~\cite{TPLink, Wilocity, SiBeam, OpenMili, qualcomm,samsung} offer mmWave systems but they take a long time to steer the beam or require complex hardware, making them unsuitable for mobile clients~\cite{60GHzMobicom,WiMi}. Also, some research on mmWave focuses on point-to-point links for Data Centers~\cite{flyways, mmw_DC2,mmw_DC3} or cellular and WiFi applications~\cite{60GHzMobicom, BeamSpy, WiMi, MOCA, BBS}.  These implementations typically use a horn antenna to direct the beam, which requires mechanical steering and is unsuitable for mobile links.

 
 \vskip 0.06in \noindent {\bf (b) Sparse Recovery Theory:} 
The theoretical problem we consider falls under "sparse phase retrieval"~\cite{iwen2017robust,jaganathan2015phase}. Generally, the goal is to recover an approximation of a $K$-sparse vector, ${\bf x}$,  from $M$ measurements of the form $|\bB {\bf x}|$. The presence of the absolute value is what makes this problem different from the usual compressive sensing~\cite{CRT,Don} and sparse FFT~\cite{gilbert2002near,GSM,SFFT1, SFFT2, gilbert2014recent}.  In our context, we have an extra restriction that the matrix $\bB$ is of the form $\bA \bF'$, where $\bF'$ is the inverse Fourier transform and all entries in $\bA$ have unit magnitude, i.e.,  $|a_{ij}|=1$ for all $i,j$. To the best of our knowledge, this form has not been considered before.

Some of our proof techniques are inspired by past work on sparse FFT,  particularly the work in~\cite{gilbert2002near,GSM} which used boxcars filters for sparse Fourier transform algorithms. However, the technical development of our proofs is different due to the leakage between multiple  beam arms, which requires extra layers of randomization.  Furthermore, our results and problem are different due to the restrictions on our measurements that do not exist in sparse FFT. 

\vskip 0.06in \noindent{\bf (c) Massive MIMO Beamforming:} 
Massive MIMO~\cite{argos, argos2} has many antennas, but the signal from each antenna can be received and manipulated independently.  In contrast, a mmWave phased array receives only the combined signal from all its antennas. Thus, beamforming techniques in massive MIMO (and standard MIMO) do not apply to mmWave phased array.

\begin{figure*}[t!]
	\centering
	\includegraphics[width=6.5in]{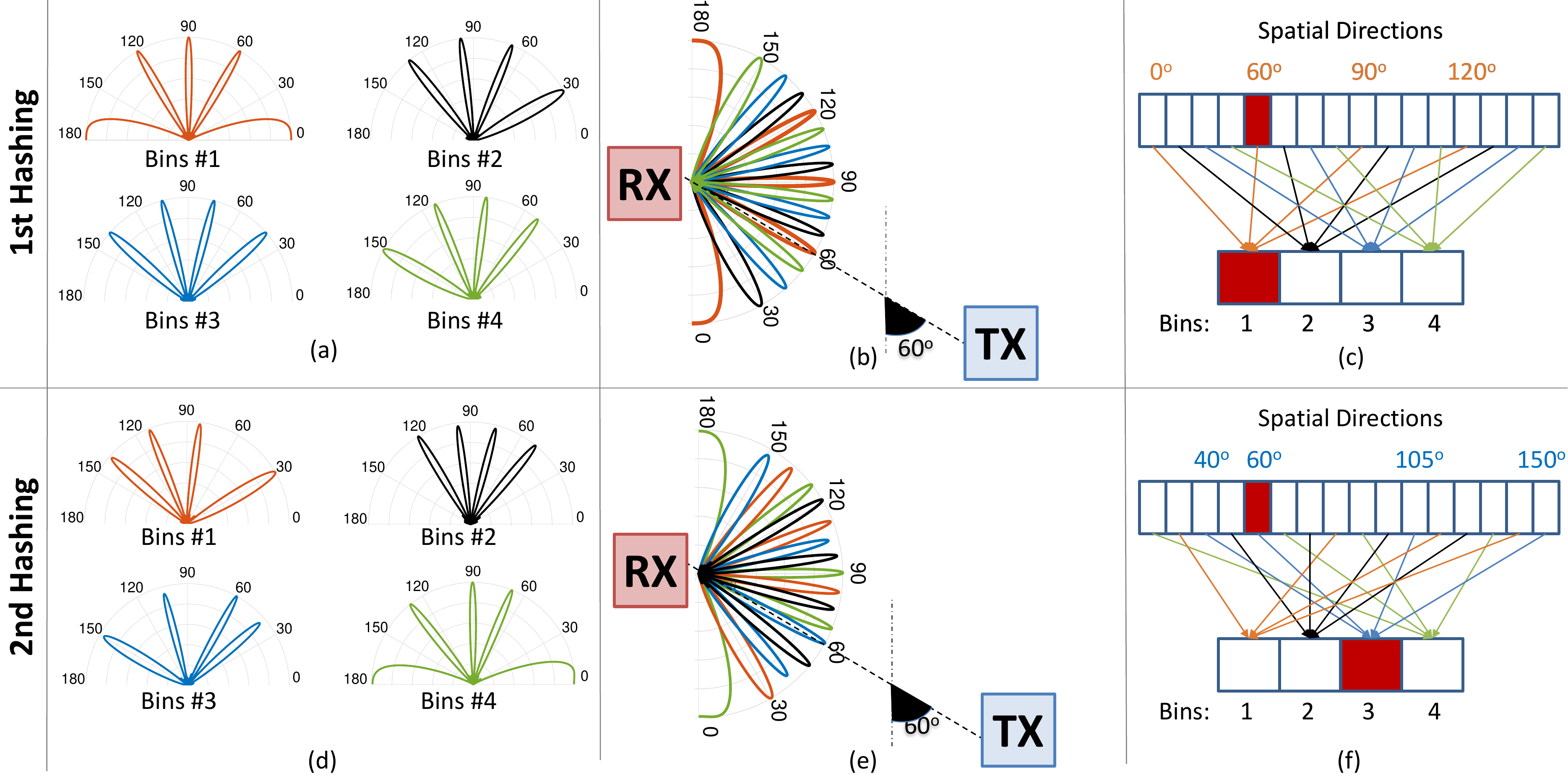}
\vskip -0.1in
\caption{\textbf{Illustrative example of \name's algorithm}}
\vskip -0.1in
\label{fig:example}
\end{figure*}

\section{Illustrative Example}\label{sec:example}
Let us start by explaining the intuition underlying our algorithm. Consider a scenario where the transmitter is at a $60^o$ angle with respect to the receiver. For clarity, assume a 2D setting. The same argument can be extended to 3D. We would like the receiver to detect that the best alignment is along the $60^o$ angle, but without scanning the space.  

\name\ avoids the need to scan all spatial directions by using multi-armed beams, which simultaneously sample the signal along multiple directions. Say for example, that there are 16 possible directions in space, i.e., $N=16$.  \name\ can sample all of these directions using 4 multi-armed beams, each covering $N/4=4$ directions in space. Fig.~\ref{fig:example}(a) shows four such multi-armed beams, and Fig.~\ref{fig:example}(b) shows how together they cover the whole space of directions.  Such set of multi-armed beams operates like a hash function, where $N=16$ directions are hashed into 4 bins, and each bin covers $N/4 = 4$ distinct directions.  The value of the bin represents the combination of the signals that hash into it. For example, if the signal is coming along the $60^\text{o}$ direction and $60^\text{o}$ hashes to bin number 1, then only bin one will have energy whereas the other bins will have no energy, as shown in Fig.~\ref{fig:example}(c).  Thus, one can ignore directions that hash to bins 2,3, and 4, and focus only on directions that hash to bin 1.  This significantly reduces the search space to the directions that hash to the first bin (i.e., the first multi-armed beam).

At this stage, we know that the signal could have come from the directions covered by the first bin i.e., $0^\text{o}$, $60^\text{o}$, $90^\text{o}$ and $120^\text{o}$. But we do not know which among them is the correct direction.  Thus, we change the hash function and try again. To do so, we use a second set of multi-armed beams which together hash the whole space of directions into a set of bins. The hash however is randomized with respect to the previous hash so that directions that got hashed together are unlikely to hash together again.  Figs.~\ref{fig:example}(d,e,f) show an example of hashing the spatial directions into bins after randomizing the multi-armed beams and hence randomizing which directions map into which bins. The first bin now collects energy along $30^\text{o}$, $80^\text{o}$, $110^\text{o}$ and $140^\text{o}$. Since the signal is arriving along $60^\text{o}$, it will be captured by the third bin which is represented in blue in Fig~\ref{fig:example}(f). Hence, in this second hashing, only the energy of the third bin will be large. This suggests that the signal arrived along one of the directions that mapped to the third bin which in this case are $40^\text{o}$, $60^\text{o}$, $105^\text{o}$ and $150^\text{o}$. Since the $60^\text{o}$ direction is the only common candidate from both the first hashing and the second hashing, \name\ picks it as the direction of the signal. Thus, \name\ is able to find the correct direction from which the transmitter's signal arrives without having to scan all possible directions.

The above example has ignored one issue: What if the signal arrives along multiple paths (i.e., multiple directions, not just one).  
In this case, signals from different directions can collide in the same bin. Since signals are waves, they can collide to cancel each other.  Thus, there is a probability that a bin may have negligible power, though it does contain the directions of real signal paths. To avoid missing actual signal directions, each time we hash, we need to randomize the hashing. \name\ repeats the hashing while randomizing the directions that fall into the same bin. This ensures that if the signals along two paths collide (and cancel each other) in one bin in the first hashing, the probability they continue to collide decreases exponentially with more hashes.

In more general settings, \name\ can use more than just two hash functions.
After each hashing, the bins that have energy vote for the directions that map
into those bins.  Thus, a direction will get a vote if it falls in a bin that
has energy.  Since mmWave signals at most have two or three
paths~\cite{rangan2014millimeter,mmw_measurement2}, the directions from which
the signal can arrive are sparse. As a result, these directions are likely to
get a lot of votes whereas other directions are unlikely to get a lot of votes.
\name\ can quickly discover the directions of the paths along which the signal
arrives by picking the direction that acquire the highest number of votes.



\section{\name}\label{sec:rapidlink}
This section describes \name\ in detail. For clarity, we describe the algorithm assuming only the
receiver has an antenna array whereas the transmitter has an omni-directional antenna. In~\xref{sec:2sides}, we extend it to the to the case where both transmitter and receiver have antenna arrays.

\subsection{Problem Statement}\label{sec:problem}
Recall that the problem is defined as follows:  Let $\bf{x}$ be a $K$-sparse $N$-element complex vector that denotes the signal along various spatial directions. The objective is to estimate the power (i.e., magnitude) of the signal along each direction,  $|x_i|$, using a small number of measurements of the form $y_j = |{\bf a}_j \bF' {\bf x}|$, where $\bF'$ is the inverse Fourier matrix, and ${\bf a}_j$  is a vector of phase shifts, $|a_{ij}|=1$, that are under our control. 

Before describing our solution to this problem, it is important to understand why the phases of the measurements are not usable. Every measurement involves sending a frame from transmitter to receiver. Since the oscillators on the transmitter and receiver always experience some CFO (Carrier Frequency Offset)~\cite{MegaMIMO}, the signal of each frame incurs an additional unknown phase shift. Further, this phase shift changes across frames. 
Correcting the CFO across measurement frames is not supported in the 802.11ad standard~\cite{802.11ad}. Furthermore, such a correction will be very hard due to the high frequencies of mmWave. For example, a small offset of 10 parts per million at such frequencies can cause a large phase misalignment in less than hundred nanoseconds. 

\subsection{\name's Algorithm}\label{sec:alg}
\name\ works in two stages. First, it randomly hashes the space into bins (using multi-armed beams) such that each bin collects power from a range of directions. Second, it uses a voting mechanism to recover the directions that have the power. Below, we describe these two stages in detail.

\vskip 0.06in \noindent{\it A. Hashing Spatial Directions into Bins}\newline
\name\ hashes the signal along various directions to bins using multi-armed beams. Let us refer to the arms in each multi-arm beam as the {\em sub-beams}.  Let $R$ be the number of sub-beams in each beam, $B$ the number of bins in each hash function, and $L$ the number of hash functions. 

Each setting of the phase shifter vector, ${\bf a}$, creates a different beam pattern and the resulting measurement $y =|{\bf a} \bF' {\bf x}|$ will correspond to the power in the directions covered by the beam pattern. 

So how do we create multi-armed beams? It should be clear that, given the structure of the measurements, we can create a beam that points in one direction, $s$, by setting $\bf a$ to the $s$-th row in the Fourier matrix. Thus, to create a multi-armed beam,  \name\ divides the vector ${\bf a}$ into $R$ segments each of length $N/R$ \ie, ${\bf a}_{[1:N/R]},{\bf a}_{[N/R+1:2N/R]},\cdots,{\bf a}_{[(R-1)N/R:N]}$. Each segment then sets its sub-beam towards a different direction. This is done by setting the segment ${\bf a}_{[i:i+N/R]}$ to the corresponding segment in the desired row of the Fourier matrix. Formally, if an index $i$ belongs to the $r$-th segment pointing towards the direction $s^r$, then $a_i = ({\bf F}_{s^r})_i  \cdot e^{-j 2\pi t_r/N}$, where $t_r$ is a random integer between $0, \ldots N-1$, and  $({\bf F}_{s^r})_i$ refers to the $i$-th entry in row $s^r$ in the Fourier matrix. The term $e^{-j 2\pi t_r/N} $ results in a phase shift of the sub-beam without changing its direction, and it simply helps in the proof.

Due to the properties of the Fourier transform, the sub-beam created by each segment will be wider than a single beam created by the full array by a factor of $R$, so each sub-beam covers $R$ adjacent directions.  Since there are $R$ such sub-beams, the multi-armed beam created by this setting of the vector ${\bf a}$ will cover $R^2$ directions. Now, if we wish to hash the space of directions into $B$ bins, then each will cover $R^2$ directions and hence $B = N/R^2$. 

But how do pick the directions of the sub-beams in each multi-armed beam?
The best scenario is when the sub-beams in each multi-armed beam are well spaced so that the leakage from their side-lobes is
minimized.  Fig~\ref{fig:beam_hashing}(a) shows an example of well-spaced sub-beams, where we have a multi-armed beam with two sub-beams directed $60^o$ apart.  In this case, the beam pattern will hash directions that are $60^o$ apart into the same bin. By shifting the direction, we can then create all different bins in the hash function, as shown in Fig~\ref{fig:beam_hashing}(b) where each color corresponds to a bin in the hash function.  Formally,  the above process is achieved by setting the direction of the $r$-th segment in bin $b$ to be equal to $s^r_b= Rb + r P$, where $P=N/R$ is the spacing between two sub-beams corresponding to the same bin.

\begin{figure}[t!]
\centering
\begin{tabular}{cc}
	\includegraphics[width=1.3in]{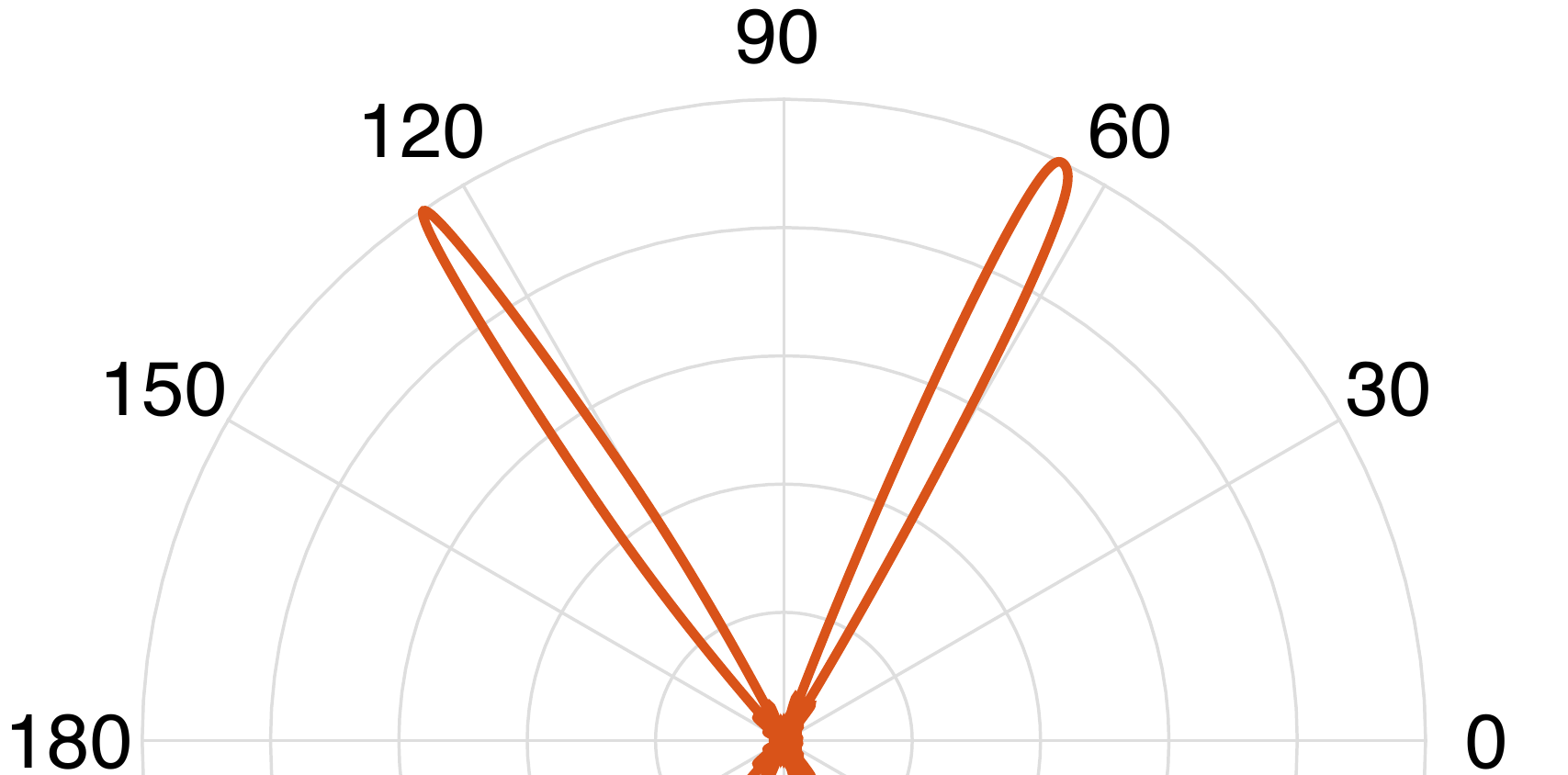} & \includegraphics[width=1.3in]{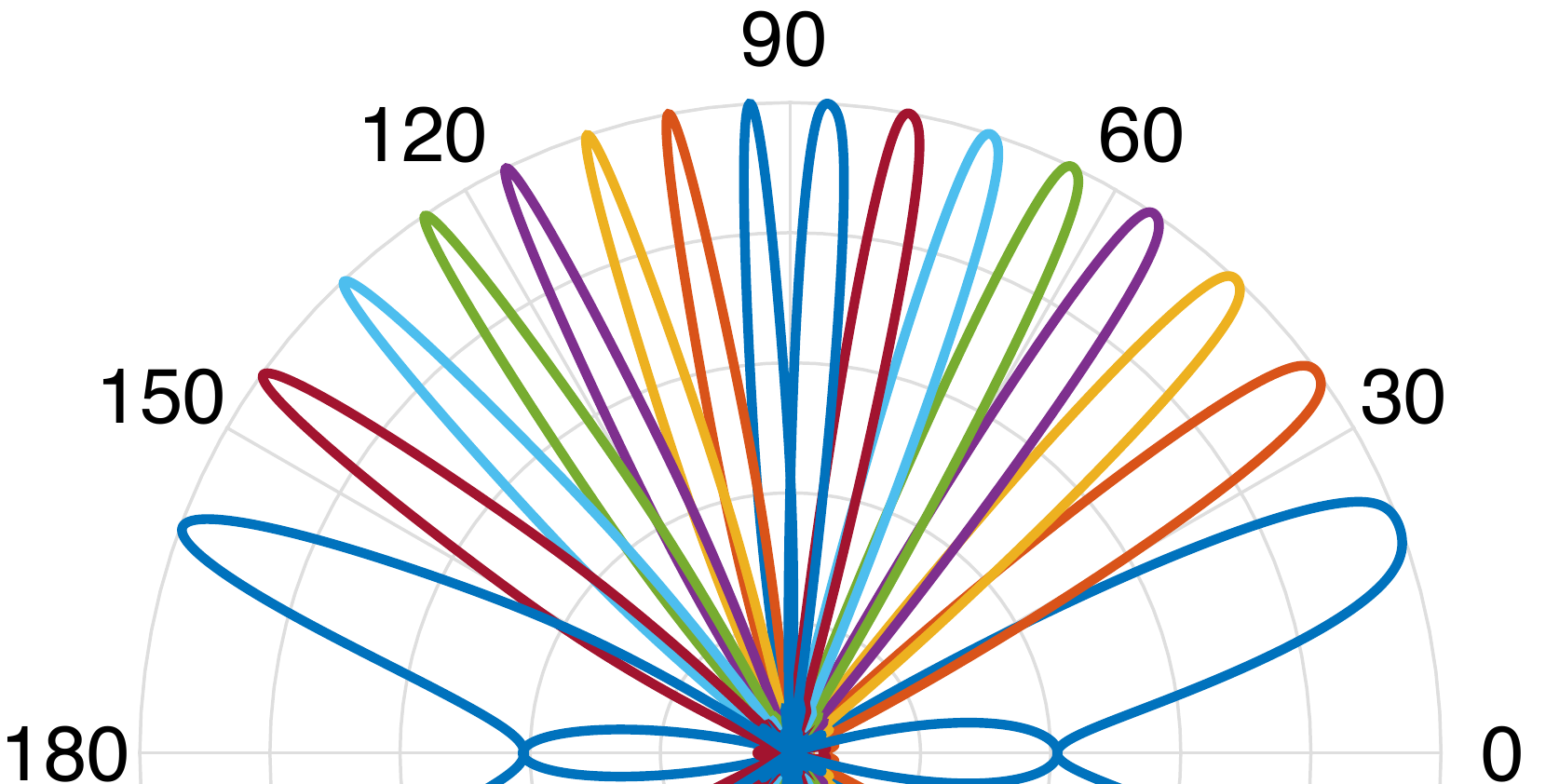}\\
(a) & (b) 
\end{tabular}
\vskip -0.1in
\caption{\footnotesize \textbf{Hashing Beam Patterns} (a) \name\ hashes well spread directions into a bin; (b) All the bins of \name's hash function where the same color corresponds to directions in the same bin, and hash function covers the whole space.} 
\vskip -0.1in
\label{fig:beam_hashing}
\end{figure}

The question that remains is how do we randomize the hashing process to make sure that two large signals are not hashed to the same bin every time.  Ideally, we can solve this problem by randomly permuting the entries in of the vector $\bf x$. Physically, however we cannot permute $\bf x$. Instead, we leverage a nice property of the Fourier transform that says that we can pseudo-randomly permute the {\em input} of the Fourier transform (i.e., the entries of the vector ${\bf x}$) by pseudo-randomly permuting its output (i.e., the entries of the vector $\bF'{\bf x}$)~\cite{gilbert2002near,GSM,SFFT2}.  In the simplest setting, suppose that we want to permute  the input samples according to:

\begin{equation}
\label{e:time}
{\bf x'}(f) = {\bf x}(\sigma f),
\end{equation}
where $\sigma$ is chosen at random for an appropriate distribution, and all operations on indices are done modulo $N$. This is equivalent to permuting the output $\wh{\bf x} = \bF' {\bf x}$ according to:
\begin{equation}
\label{e:freq}
\wh{\bf x'}(t) = \wh{\bf x}(\sigma^{-1} t ).
\end{equation}
where $\sigma^{-1}$ is the inverse of $\sigma$ modulo $N$.  This enables us to ``randomize'' the positions of large entries in ${\bf x}$ by rearranging the entries of $\bF' {\bf x}$.  This result is very useful since we can easily permute $\bF' {\bf x}$ in our measurements by simply permuting the elements in the vector $\bf a$, i.e., by permuting the phase shifts applied to the phase shifters on each antenna. 

More generally,  let us use the vector ${\bf y}_{1\times B} = |\bA_{B\times N} \bF' {\bf x}|$ to refer to the $B$ measurements performed as part of a hash function, where $\bA$ is a matrix of phase shifts. To permute this hash function, we measure ${\bf y}_{1\times B} = |\bA \bP' \bF' {\bf x}|$,  where we use the term $|.|$ to refer to the magnitude of the individual elements in the vector, and the matrix $\bP'$ is a generalized permutation matrix.\footnote{
$\bP'$ is a generalized permutation matrix if each row or column of $\bP'$ contains exactly one non-zero entry, and that entry has unit magnitude. We use the matrix $\bP'$ defined as in~\cite{SFFT2}:  the $i$-th column of $\bP'$ contains the value $\omega^{a\sigma i}$ in the row $\sigma(i-b)$, where $\omega=e^{2 \pi j/N}$ and $a,b,\sigma$ are randomly chosen parameters.  This has the effect of rearranging the vector ${\bf x}$ by moving the entry $x_i$ to $x_{\rho(i)}$ for $\rho(i)=\sigma^{-1} i +a$, and multiplying it by $\omega^{bj+\sigma b a}$. Note that the role of the latter multiplier is similar to that of the phase multiplier described in~\xref{sec:alg}(A), namely it helps with the proof. See the appendix for a more detailed description of the process.}
That is, our phase shifter matrix is equal to $\bA \bP'$ (note that each entry of the latter matrix has a unit magnitude, i.e., it represents proper phase shifts).
This is equivalent to measuring $|\bA \bF' \bP {\bf x}|$, for the corresponding generalized permutation matrix $\bP$.

\vskip 0.06in \noindent{\it B. Recovering the Directions of the Actual Paths}\newline
After hashing the spatial directions into bins, \name\ discovers the
actual directions of the signal using a voting scheme where each
bin gives votes to all directions that hash into that bin. After few
random hashes, the directions that have energy will collect the
largest number of votes which allows \name\ to recover them.
Unfortunately, simply directly applying this voting approach does not work
well because the side-lobes of the beams create
leakage between the bins and hence a strong path in one bin can leak
energy into other bins which corrupts the voting process. To overcome
this problem, \name\ uses a form of soft voting that takes into account
the leakage between the bins.

Specifically, \name\ models the beam patterns (example of which are shown in Fig.~\ref{fig:beam_hashing}(a) and Figs.~\ref{fig:example}(a) and (d),) as a {\em coverage function} $I(b,\rho,i)$ that indicates the coverage of the direction $i$ by the beam corresponding to the bin $b$, assuming the indices are permuted by $\rho$. It is defined as: 
\begin{equation}
	I(b,\rho, i) = |{\bf a}^b {\bf F'}_{\rho(i)} |^2
\end{equation} 
where ${\bf a}^b$ is the vector defining the settings of the phase shifters corresponding to bin $b$.

If we hash into $B$ bins, we will have $B$ such
patterns and collect $B$ measurements ${\bf y}_{1\times B}$
corresponding to the bins. After taking the magnitude squared of each
measurement,  we can estimate the energy of  the signal coming from direction
$i$ as:
\begin{equation}
\label{e:t}
T(i,\rho) = \sum_{b=0}^{B-1} y_b^2\times I(b,\rho,i), 
\end{equation} 
If the estimate $T(i,\rho)$ exceeds a predefined threshold $T$ then we conclude that there is a signal coming from direction $i$.  Otherwise we conclude that there is no such signal. 

\vskip 0.06in \noindent{\it C. Performance Analysis}\newline
A detailed performance analysis is provided in the Appendix.  Here we note the main theorem and its implications. 
\begin{theorem}
\label{t:main}
Suppose that the vector ${\bf x}$ has at most $K$ non-zero entries, at that the energy $x_i^2$ of each non-zero entry $x_i$ is at least $1/K$. Furthermore, suppose that $N$ is a prime. There exists a setting of parameters $T$, $R$ and $B=O(K)$, so that for each candidate direction $i$,  we have that:
\begin{Itemize}
\item If $x_i \neq 0$ then $T(i,\rho) \ge T$ with probability at least $2/3$
\item If $x_i=0$ then $T(i,\rho) < T$   with probability at least $2/3$
\end{Itemize}
\end{theorem}

The probability of correctness can be amplified by performing several ($L$) random hashes and aggregating the results. There are multiple ways of performing the aggregation. A simple approach is to use ``hard voting'', i.e., conclude that a signal is coming from a direction $i$ if the signal from that direction has been detected by the majority of hashes. By Chernoff bound, this approach reduces the probability of incorrect detection from $1/3$ to $e^{-C'' L}$ for some constant $C''$. By letting $L=O(\log N)$, we can compute correct estimates for all  {\em all} indices $i$, with the probability  of failure at most $1/N$. The algorithm uses $BL=O(K \log N)$ measurements and its running time complexity is $O(N K \log N)$. 

The algorithm can be also used to provide estimations of the values of $|x_i|^2$'s. The guarantees are provided by the following theorem. Note that no assumptions about the sparsity of ${\bf x}$ are required,  although the guarantees are meaningful only for $x_i$'s whose magnitude is large enough. This makes the estimate resilient to the presence of small amounts of noise at all coordinates. 

\begin{theorem}
\label{t:main2}
Suppose that $N$ is a prime. There exists a setting of parameters $R$, $B=O(K)$, and a constant $C>1$ so that for each candidate direction $i$,  we have:
\[ \Pr[  |x_i|^2/C - \|{\bf x}\|_2^2/K \le T(i,\rho) \le  C |x_i|^2 + \|{\bf x}\|_2^2/K ] \ge 2/3 \]
\end{theorem}



\subsection{Antenna Arrays on Both Transmitter and Receiver}\label{sec:2sides}
We now extend the above model and algorithm to the case where both transmitter and receiver
have antenna arrays. In this case, each measurement can be written as:
\begin{equation}
y_{1\times 1} =| {\bf a}^{rx}_{1\times N}\bF'{\bf x}^{rx}_{N\times 1}{\bf x}^{tx}_{1\times N}\bF' {\bf a}^{tx}_{N \times 1}|
\end{equation} 
where ${\bf a}^{rx}$ and ${\bf a}^{tx}$ are vectors corresponding to a setting of the phase shifters, $ \bF'$ is an $N\times N$ inverse Fourier transform
matrix and ${\bf x}^{rx}$ and ${\bf x}^{tx}$ are sparse vectors representing the angle of arrival at the receiver and the angle of departure at the transmitter respectively. Our goal is to recover  ${\bf x}^{rx}$ and ${\bf x}^{tx}$ from several measurements $y$.

We can reduce this problem to the problem solved in the earlier section. Specifically, we make  $B\times B$ measurements of the form:
\begin{equation}
{\bf Y}_{B\times B} = |{\bf A}^{rx}_{B\times N}\bF'{\bf x}^{rx}_{N\times 1}{\bf x}^{tx}_{1\times N}\bF' {\bf A}^{tx}_{N \times B}|
\end{equation} 
where ${\bf A}^{rx}$ is the phase shift matrix as in Theorem~\ref{t:main}, and ${\bf A}^{tx}$ is its transpose. Let ${\bf A}^{rx}_i$ be the $i$-th column of  ${\bf A}^{rx}$ and ${\bf A}^{tx}_j$ be the $j$-th row of ${\bf A}^{tx}_j$.  Furthermore, let $y_i= \sum_j Y_{i,j}$. We observe that
\begin{eqnarray*}
y_i & = & \sum_j Y_{i,j} \\
& = & \sum_j   |{\bf A}^{rx}_i \bF'{\bf x}^{rx} {\bf x}^{tx} \bF' {\bf A}^{tx}_j| \\
& = & \sum_j  |{\bf A}^{rx}_i \bF'{\bf x}^{rx}| |{\bf x}^{tx} \bF' {\bf A}^{tx}_j| \\
& = & |{\bf A}^{rx}_i \bF'{\bf x}^{rx}|  ( \sum_j |{\bf x}^{tx} \bF' {\bf A}^{tx}_j|) \\
& = & |{\bf A}^{rx}_i \bF'{\bf x}^{rx}| C
\end{eqnarray*}
where $C$ is a constant independent of $i$. Therefore, we can recreate the measurements of the form needed by Theorem~\ref{t:main2} from the $B^2$ measurements provided by the matrix ${\bf Y}$. In this way we can test, for each $i$, whether ${x}^{rx}_i$ is non-zero, and each test is correct with probability at least $2/3$. By repeating the process $L=O(\log N)$ times, we can detect each non-zero entry with high probability. The total number of measurements is $B^2 L = O(K^2 \log N)$.

Finally, while we described the algorithm for 1D antenna arrays,
the algorithm holds for 2D arrays as well. We simply need to apply the hash
function along both dimensions of the array. For an $N\times N$ antenna array,
the complexity will be $O(K^2\log{N^2})$ and hence will continue to
scale logarithmically with the number of antennas in the array.

\subsection{Discussion}
In practice randomizing the phases in the measurements and permutation matrices is not necessary, and therefore in our implementation we set $b=0$ and $t_r=0$. We also drop the assumption that $N$ is prime. Finally, we use soft voting instead of the hard voting. Specifically,  we estimate of the strength of the signal along direction $i$ as $ S (i)=\prod_{l=1}^{L} T_l(i, \rho)$, where $T_l(i, \rho)$ is defined as in Equation~\ref{e:t}, for the $l$-th permutation. The soft voting approach uses more information about the measurements than hard voting, and hence its practical performance is better. We extract the significant coefficients of the signal $x_i$ by selecting the indices $i$ with the largest values of $S(i)$.


\section{\name\ Implementation}
\label{sec:platform}
We have implemented \name\ by designing and building a full-fledged mmWave
radio capable of fast beam steering, which is shown in Fig.\ref{fig:device}. The
radio operates in the new 24GHz ISM band and serves as a daughterboard for the
USRP software radios. Its physical layer supports a full OFDM stack up to 256
QAM. Our implementation addresses critical system and design issues that are
described below.

\begin{figure}[t!]
\centering
	\includegraphics[width=2.6in]{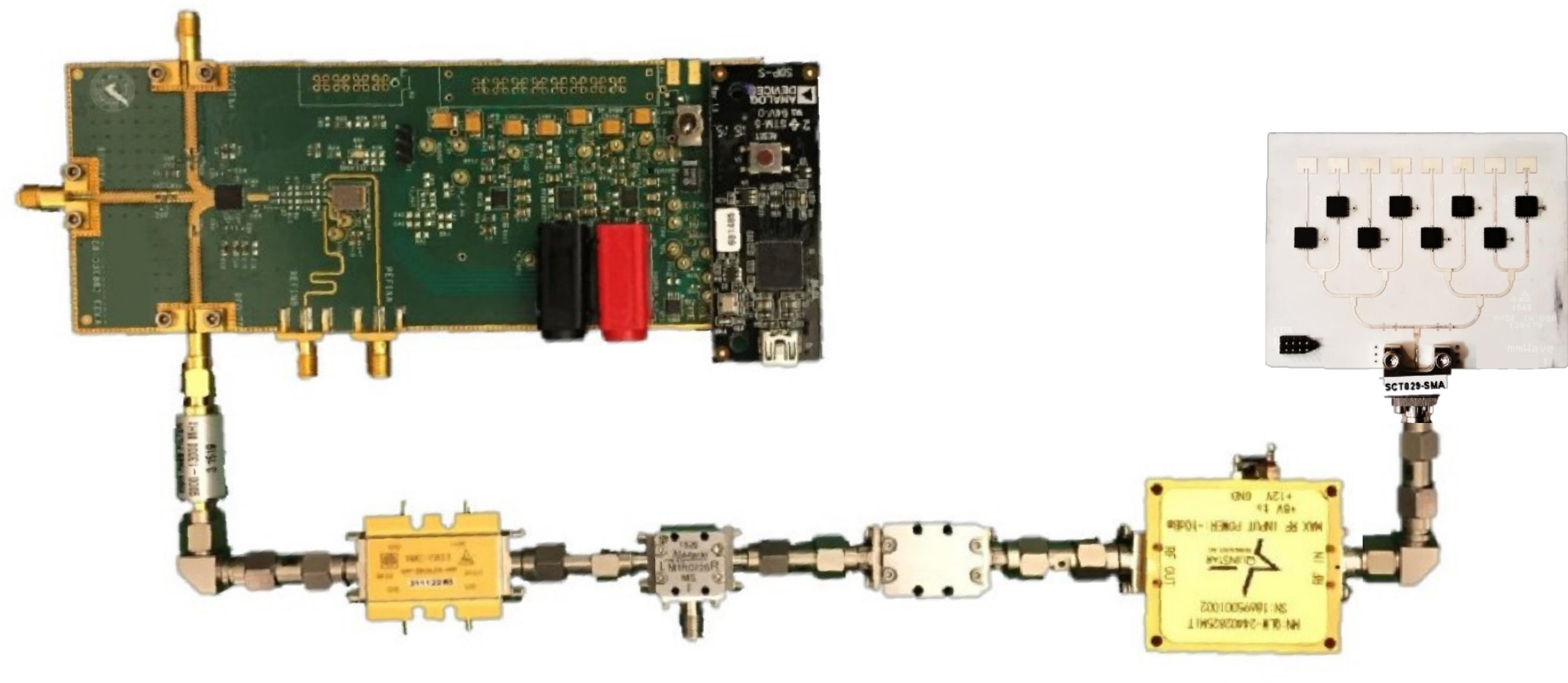}
\vskip -0.1in
\caption{ \textbf{\name\ Platform:} \footnotesize The figure shows the phased array and mmWave radio we built to operate as a daughterboard for the USRP software radio.}
\label{fig:device}
\end{figure}

\begin{figure}[t!]
\begin{tabular}{c}
	\includegraphics[width=3.2in]{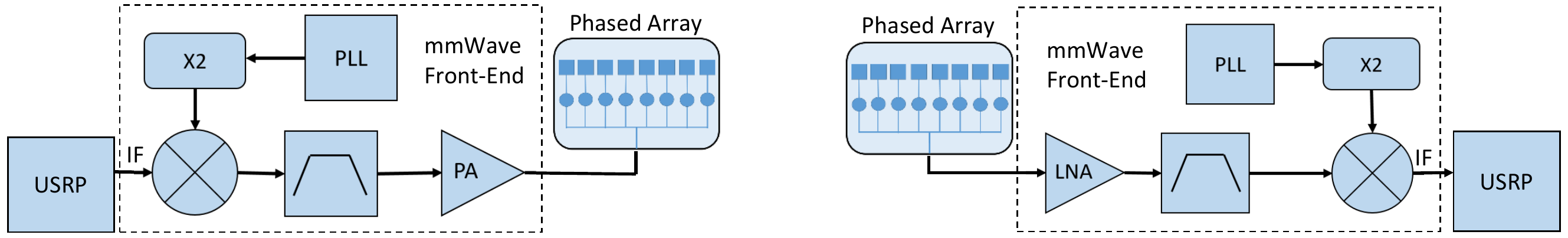}\\
\hspace{0.2in}(a)  \hspace{1in}(b)\\
\end{tabular}
\vskip -0.1in
\caption{ \textbf{\name's Architecture.} \footnotesize The figure shows 
block diagrams for both \name's transmitter(a) and receiver(b).}
\vskip -0.1in
\label{fig:BlockDiagram}
\end{figure}

\vskip 0.06in \noindent  \textbf{(a) Heterodyne Architecture:} 
Millimeter Wave hardware is significantly more expensive than GHz hardware. Thus, we leverage a heterodyne
architecture where the mmWave signal is first taken into an
intermediate frequency of a few GHz, before the I and Q (real \&
imaginary) components are separated. Such a design reduces the number
of components that need to operate at very high frequencies (e.g.,
mixers, filters, etc) and replaces them with components that operate
at a few GHz, which are much cheaper.

The architecture of \name's receiver is shown in
Fig~\ref{fig:BlockDiagram}(b). The first block is a mmWave phased array
which allow us to steer the beam electronically. The array
consists of antenna elements where each element is connected to a phase
shifter. The outputs of the phase shifters are combined and fed
to a single mmWave front-end. The front-end has a standard design of a low-noise amplifier (LNA), band-pass filter, mixer, and a PLL. 
The mmWave front-end down-converts
the mmWave signal to an intermediate frequency (IF) and feeds it to
the daughterboard on the USRP which samples it and passes the digitized
samples to the UHD driver.  This enables easy manipulation of mmWave
signals using standard GNU-radio software and allows us to build an 
OFDM stack that supports up to 256 QAM. 


We have built the design in
Fig.~\ref{fig:BlockDiagram} using off-the-shelf components. For the
mmWave low-noise amplifier (LNA) and power amplifier (PA), we use
Hittite HMC-C020 and Quinstar QLW-2440, respectively. For the mmWave
mixer, we use Marki M1R-0726MS. To generate local oscillator (LO)
signals, we use Analog Devices ADF5355 PLL and Hittite HMC-C035
frequency doubler. The
phased array includes 8 antenna elements separated by
$\frac{\lambda}{2}$, where each element is connected to a Hittite
HMC-933 analog phase shifter. We use Analog Device AD7228
digital-to-analog converters (DAC) and Arduino Due micro-controller
board to digitally control the phase shifters.

\begin{figure}[t!]
\centering
	\includegraphics[width=2.8in]{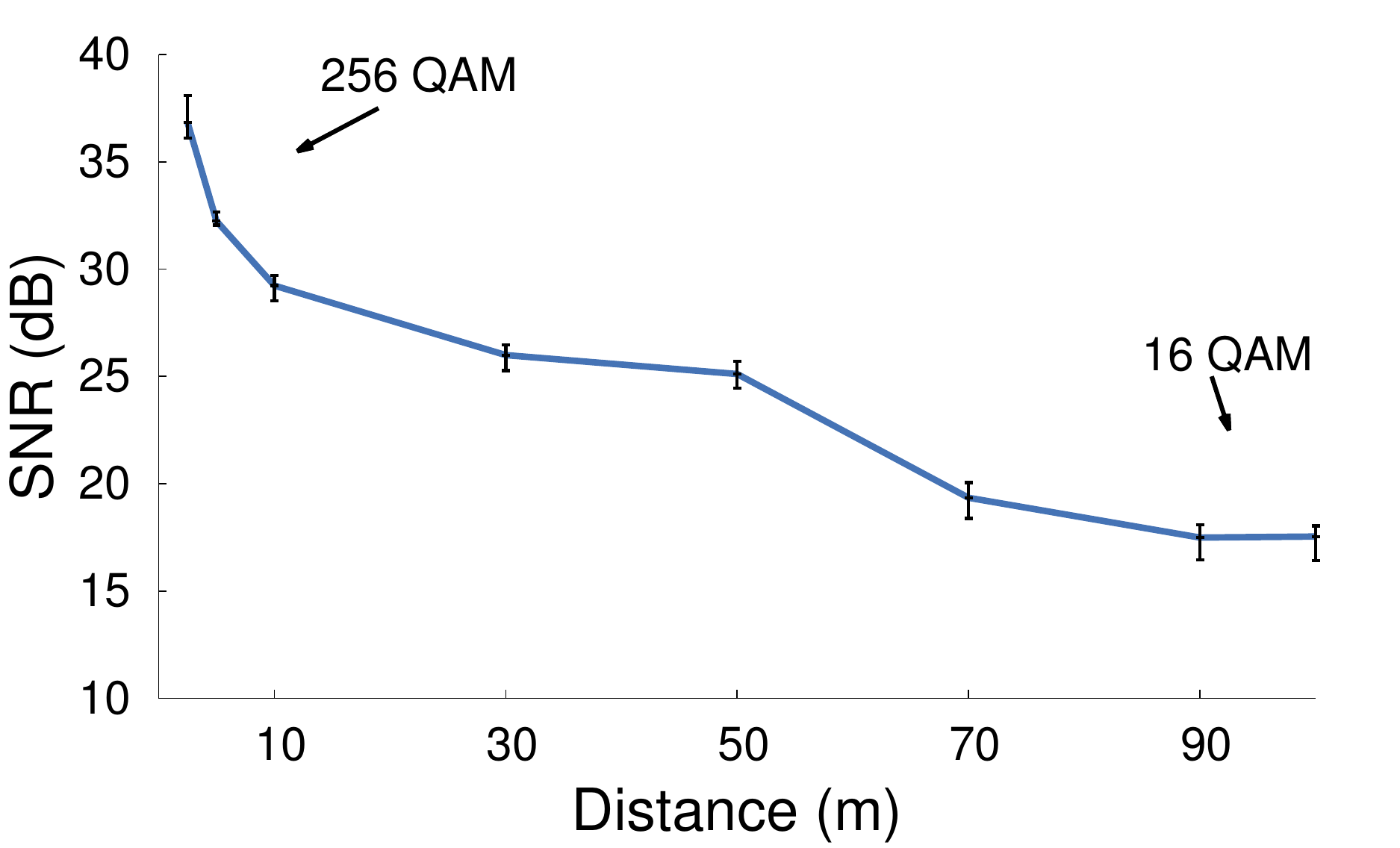}
\vskip -0.1in
\caption{ \textbf{\name\ Coverage.} \footnotesize SNR at the receiver versus distance between the receiver and the transmitter.}
\vskip -0.15in
\label{fig:snrvsrange}
\end{figure}



\vskip 0.06in \noindent  \textbf{(b) Radio Performance:} 
To test \name's ability to deliver high data rates and long range using phased
arrays, we measure the SNR of the received signal for different
distances between \name's receiver and transmitter where the transmit power
complies with FCC part15. Fig.~\ref{fig:snrvsrange} shows the
SNR at the receiver versus the distance between
transmitter and receiver. The figure
shows that \name's implementation provides SNR of more than 30~dB for distances
smaller than 10m and 17~dB even at 100m which is sufficient for relatively dense modulations such as
16 QAM~\cite{TseV:05}.


\section{Experimental Evaluation}\label{sec:results}

We evaluate \name's ability to identify the best beam
alignment quickly and accurately. We ran experiments in a lab
area with standard furniture (desks, chairs, computers, etc.). We also ran
experiments in an anechoic chamber, where we can accurately measure
the ground truth. The anechoic chamber walls are covered with RF
absorbers to eliminate multipath and isolate the space from exterior
interference. This isolation is necessary to measure the ground truth
path traveled by the signal without having RF reflections.  

\subsection{Compared Schemes}
\label{sec:baselines}
We compare the following three schemes:
\begin{Itemize}
\item {\bf Exhaustive Search:} In this approach, for each setting of the transmitter's beam direction, 
the receiver scans all different directions.  The combination of transmitter and receiver beams that delivered the maximum
power is picked as the direction of the signal. 
\item {\bf 802.11ad Standard:} 
The standard has three phases\cite{802.11ad}. The first stage is called {\it Sector Level Sweep (SLS)}. 
In this stage, the AP transmits in all possible directions, and the client sets its receiver beam pattern to a quasi-omnidirectional beam. The process is then repeated with the AP setting its receiver antenna to quasi-omnidirectional and the client sweeping through all transmit directions. At the end of this stage, the AP and client each pick the $\gamma$ directions that deliver the largest power.  The second stage is called {\it Multiple sector ID Detection (MID)}. This stage repeats the process above but with the transmit beam set to quasi-omni-directional and the scan being performed with the receive beam. This stage compensates for imperfections in the quasi omnidirectional beams. The third stage is called {\it Beam Combining (BC)}. In this stage, each of the $\gamma$ best directions at the AP are tried with each of the $\gamma$ directions at the client. Hence, $\gamma^2$ combinations are tested and the combination of transmit and receive beam directions that deliver the maximum power is then selected and used for beamforming during the data transmission. In our experiments, we set $\gamma = 4$. 
\item {\bf \name:} We run the algorithm described in~\xref{sec:alg}, where the total number of measurements is set to $K^2\log{N}$, where $N$ is the number of possible signal directions.  We set $K$ to 4 since past measurement studies show that in mmWave frequencies the channel has only 2 to 3 paths~\cite{rangan2014millimeter,mmw_measurement2,WiMi,BeamSpy}.
\end{Itemize}

\subsection{Evaluation Metrics}

We evaluate the  performance of \name's beam searching algorithm along two axes. The first is the accuracy in detecting the best
alignment of the receiver's and transmitter's beams. In this case, our metric is the SNR loss in comparison to the optimal alignment, \ie, how much SNR could we have gained had we known the ground truth. We calculate this metric by measuring 
the SNR achieved by our beam alignment and subtract it from the SNR achieved by the optimal alignment.
\begin{equation}
SNR_{loss} = SNR_{optimal} - SNR_{\name}
\end{equation}
The lower the SNR loss, the higher is our accuracy in detecting the direction
of the signal.  In order to measure the optimal SNR, we ran experiments in an
anechoic chamber where there is no multipath and we can accurately measure the
ground truth direction of the signal and align the beams along those
directions. However, we also ran experiments in multi-path rich environments.
In this case, since we do not know the ground truth, we compute the SNR loss
metric relative to the exhaustive search baseline described above. 
\begin{equation}
SNR_{loss} = SNR_{Exhaustive} - SNR_{\name}
\end{equation} 

The second metric is the latency in identifying the correct beam alignment in comparison to exhaustive search and the 802.11ad standard.

\subsection{Beam Alignment Accuracy vs. the Ground Truth}

As described above, we first run the experiments in an anechoic chamber, where
 we can accurately measure the ground
truth. For each experiment, we place \name's transmitter and receiver at two
different locations. We then change the orientation of the transmitter's and
receiver's antenna arrays with respect to each other.  Since there
is only a single line-of-sight path in the anechoic chamber, this path will
appear at a different direction at the transmitter and at the receiver depending
on the orientation of the antenna arrays. Hence, this allows us to test any
combination of directions from which the strongest path can leave the
transmitter and arrive at the receiver. For each setting, the transmitter 
transmits measurement frames (as required in 802.11ad) which the receiver uses to 
compute the directions of the best beam alignment. We then
steer the beams based on the output of the alignment and measure the SNR
achieved by this alignment. 


We first show the best alignment detected by \name\ for a particular scenario, where we set the transmitter to be at an angle of
$120^\text{o}$ with respect to the receiver's array, and the receiver at an angle of $50^\text{o}$ with respect to the transmitter's array.
Fig.~\ref{fig:beamres} shows the signal direction recovered by \name\ on the transmitter's and receiver's sides for this scenario. 
As can be seen in the figure, \name\ has accurately identified the direction of arrival of the signal at the receiver as $120^\text{o}$, and the direction of departure of the signal from the transmitter as $50^\text{o}$. 

\begin{figure}[t!]
\centering
\begin{tabular}{cc}
	\includegraphics[width=1.5in]{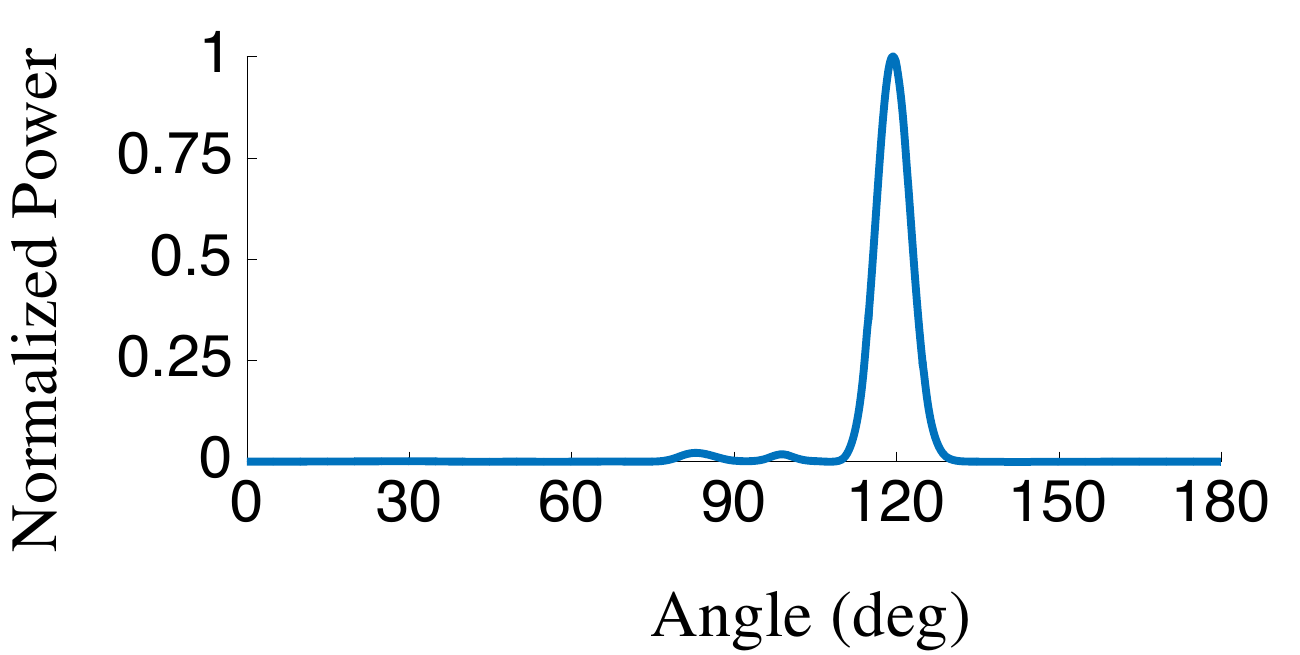} &
	\includegraphics[width=1.5in]{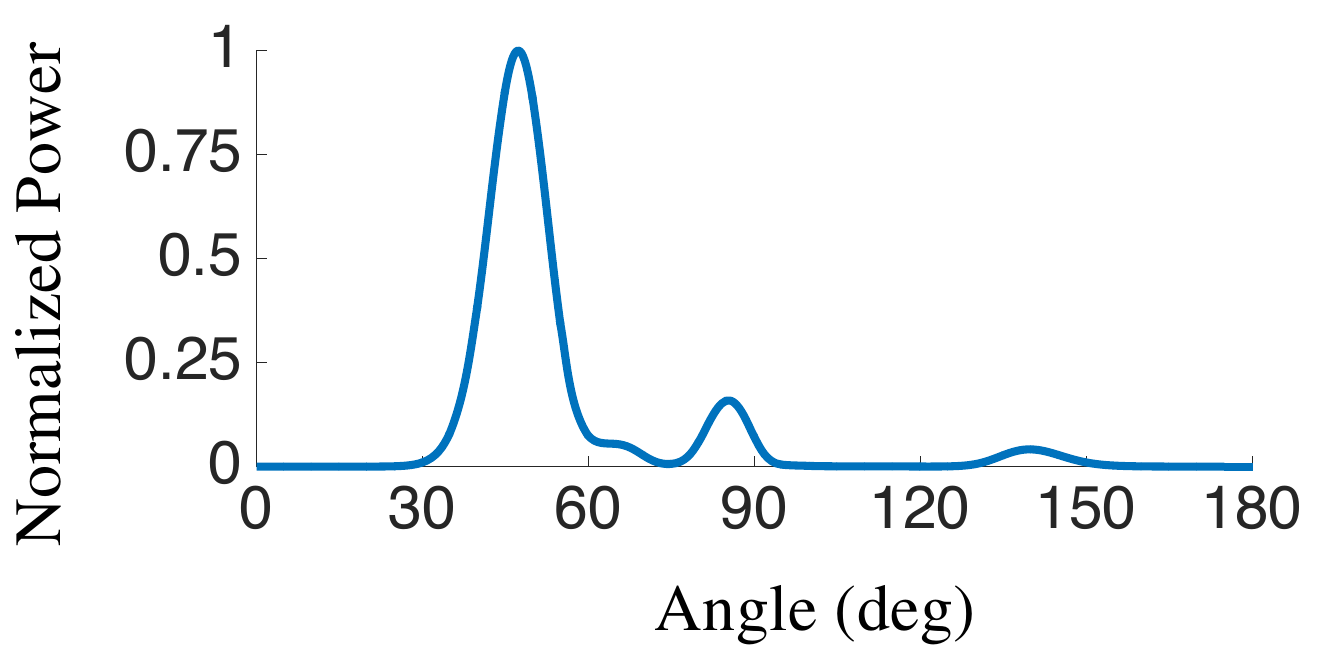}\\
\hspace{0.4in}(a)&
\hspace{0.4in}(b)
\end{tabular}
\vskip -0.1in
\caption{ \textbf{Recovered Directions} \footnotesize The figure shows (a) the direction of arrival that \name\ recovers at the receiver and (b) the direction of departure that it recovers at the transmitter when the transmitter is at a direction of $120^\text{o}$ relative to the receiver's array, and the receiver is at a direction of $50^\text{o}$ relative to the trasnmitter.} 
\vskip -0.1in
\label{fig:beamres}
\end{figure}


Next, we compare \name's algorithm with the exhaustive search and the 802.11ad standard. To do so, we repeat the above experiment by changing the orientation of the transmitter's and receiver's antenna arrays with respect to each other. We try all possible combination of the antennas' orientations, by rotating each antenna array with respect to the other for all angles between $50^\text{o}$ and $130^\text{o}$ with increments of $10^\text{o}$.
Fig. \ref{fig:cdf_SNRloss_measure} plots a CDF of the SNR loss for
\name's beam searching scheme, the exhaustive search and the 802.11ad
standard, in comparison to the optimal alignment. The figure reveals two interesting points:
\begin{Itemize}
\item The figure shows that \name\ performs better than the two baselines in
that it has minimal SNR loss.  While all schemes have a median SNR loss
below 1dB, the $90^{th}$ percentile SNR loss for both exhaustive search and the
standard is 3.95dB which is higher than the 1.89dB SNR loss of \name. This is due to the fact that
the standard and exhaustive search choose to steer using the best beam from a discrete
set of $N$ beams which they tested. However, the space of beam directions is continuous and the best beam may not exactly align with the discretization chosen by the algorithms.  In this case, they will end up picking the closest beam in the discrete set, which may not
be the exact optimal one. SNR loss is further exacerbated by the fact that this can
happen on both sides \ie, the transmitter and the receiver. In contrast, \name\
uses the beams as a continuous weight over the possible choice of directions (Equation~\ref{e:t})
 and picks the direction that maximizes the overall weight, as described
in~\xref{sec:alg}.  Thus, \name\ can discover the direction of the path beyond
the $N$ directions used by exhaustive search and the standard. 
    
\item The figure also shows that the standard and exhaustive search have similar
performance. This might seem surprising since one may expect exhaustive search
to find a better beam alignment since it spends more time searching the space.
However, it is important to recall that the standard differs from the
exhaustive search only in the first stage where it uses a quasi-omnidirectional
beams to limit the search space to a few top candidates. In the final the stage,
however, the standard tries all possible combinations of these candidate beams.
Since there is only one path in this experiment, as long as the best beam is picked 
as one of the candidate beams in the first stage, the standard will converge to
the same beam alignment as the exhaustive search. However, we will show next that
this does not continue to hold in multipath settings. 
\end{Itemize}  

\begin{figure}[t!]
\centering
	\includegraphics[width=2.8in]{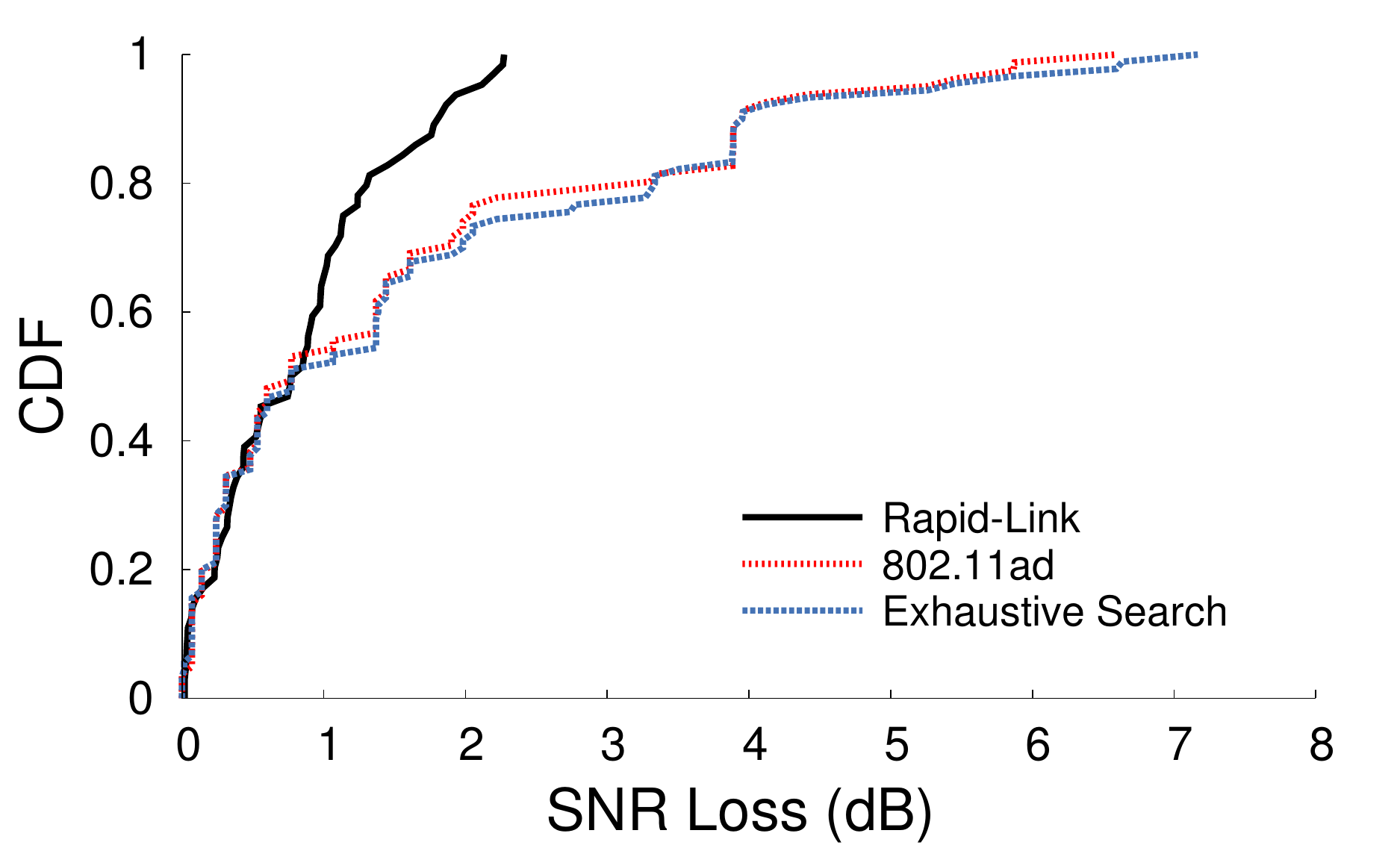}
\vskip -0.1in
\caption{ \textbf{Beam Accuracy with a Single Path} \footnotesize The figure
  shows the SNR loss due to beam misalignment for \name, the 802.11ad standard, and
  exhaustive search.}
\vskip -0.1in
\label{fig:cdf_SNRloss_measure}
\end{figure}

\subsection{Alignment Accuracy in Multipath Environments}
We repeat the above experiments in an office environment, where due to multipath, the signal can arrive
from different directions. However, as described earlier, in this case, we do not have the ground truth for the direction of strongest path and hence we measure the SNR loss relative to the exhaustive search baseline. Note that
since exhaustive search exhaustively tries all possible combinations of directions, it is
not sensitive to multipath and maintains its performance.

Fig. \ref{fig:cdf_SNRloss_multipath} plots a CDF of SNR loss for \name\ and the 802.11ad standard
with respect to the exhaustive search. The figure shows that the
standard performs much worse in multipath scenarios. Specifically,
instead of having a similar SNR to the exhaustive search as before,
the median and $90^{th}$ percentile SNR loss (with respect to exhaustive
search) are 4dB and 12.5dB, respectively. This is because the
standard is using its phased array as a quasi-omnidirectional antenna
and hence the multiple paths can combine destructively, in which case the information is lost. 
Further, due to imperfections in the quasi-omnidirectional patterns, some
paths can get attenuated and hence the standard can easily choose the wrong
direction to align its beam.  In contrast, \name\ performs well even in the
presence of multipath. Specifically, the median and $90^{th}$ percentile
SNR loss with respect to exhaustive search are 0.1dB and 2.4dB, respectively.
Finally, the figure also shows that sometimes \name's SNR loss with respect 
to exhaustive search is negative. This is because in some cases, \name\ performs better that
exhaustive search for the same reasons described above.

\begin{figure}[t!]
\centering
	\includegraphics[width=2.8in]{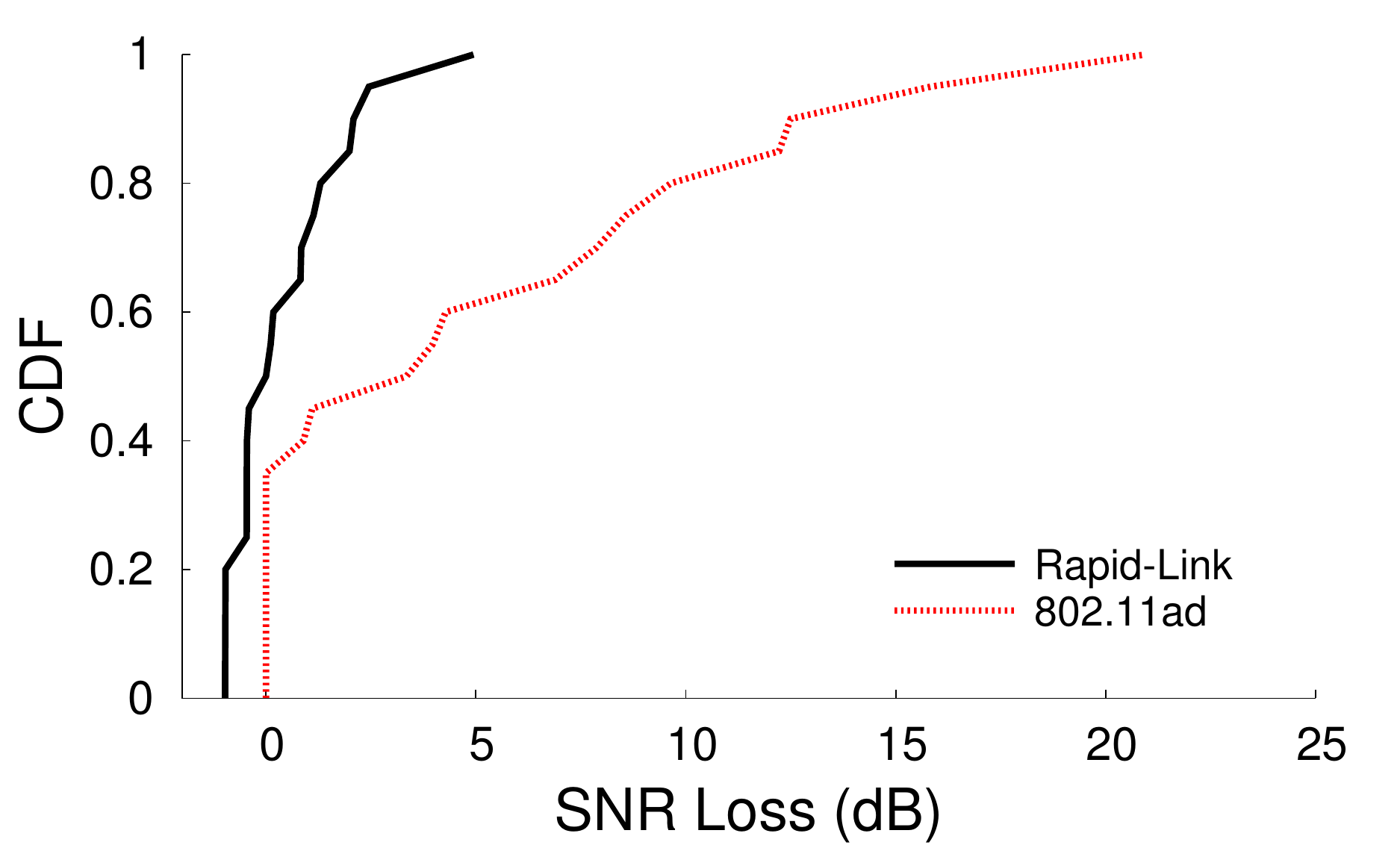}
\vskip -0.1in
\caption{ \textbf{Beam Accuracy with Multipath} \footnotesize The figure
  shows the SNR loss of the 802.11ad standard and \name\ with respect to the exhaustive search.}
\vskip -0.1in
\label{fig:cdf_SNRloss_multipath}
\end{figure}

\subsection{Beam Alignment Latency}
\label{sec:latency}
Next we would like to evaluate the gain in reducing latency that \name\
delivers over the two baselines. However, since our radio has a fixed 
array size we cannot empirically measure how this gain scales for larger arrays.
Hence, we perform extensive simulations to compute this gain for larger arrays
and we use the empirical results from our 8-antenna array to find
the delay for this array size.

\vskip 0.06in\noindent
{\it (a) Reduction in the Number of Measurements:} Since each measurement in 802.11ad requires sending a special frame, one way to measure delay is in terms of the number of measurements frames.  
Fig.~\ref{fig:searchtime_sim} plots the reduction in the number of measurements that \name\
achieves over exhaustive search and the standard. The figure shows that, for an
8-antenna phased array, \name\ can reduce the number of measurements by $7\times$ and
$1.5\times$ compared to exhaustive search and standard respectively. Further, the gain
 increases quickly as the number of antennas increase. This is due to the scaling property of each algorithm and whether it is quadratic, linear, or logarithmic. The figure shows that the gain of \name\ over exhaustive search and the standard increases very fast and for
arrays of size 256 is 16.4$\times$ better than the standard and three orders of
magnitude better than exhaustive search. 

\begin{figure}[t!]
\centering
	\includegraphics[width=2.8in]{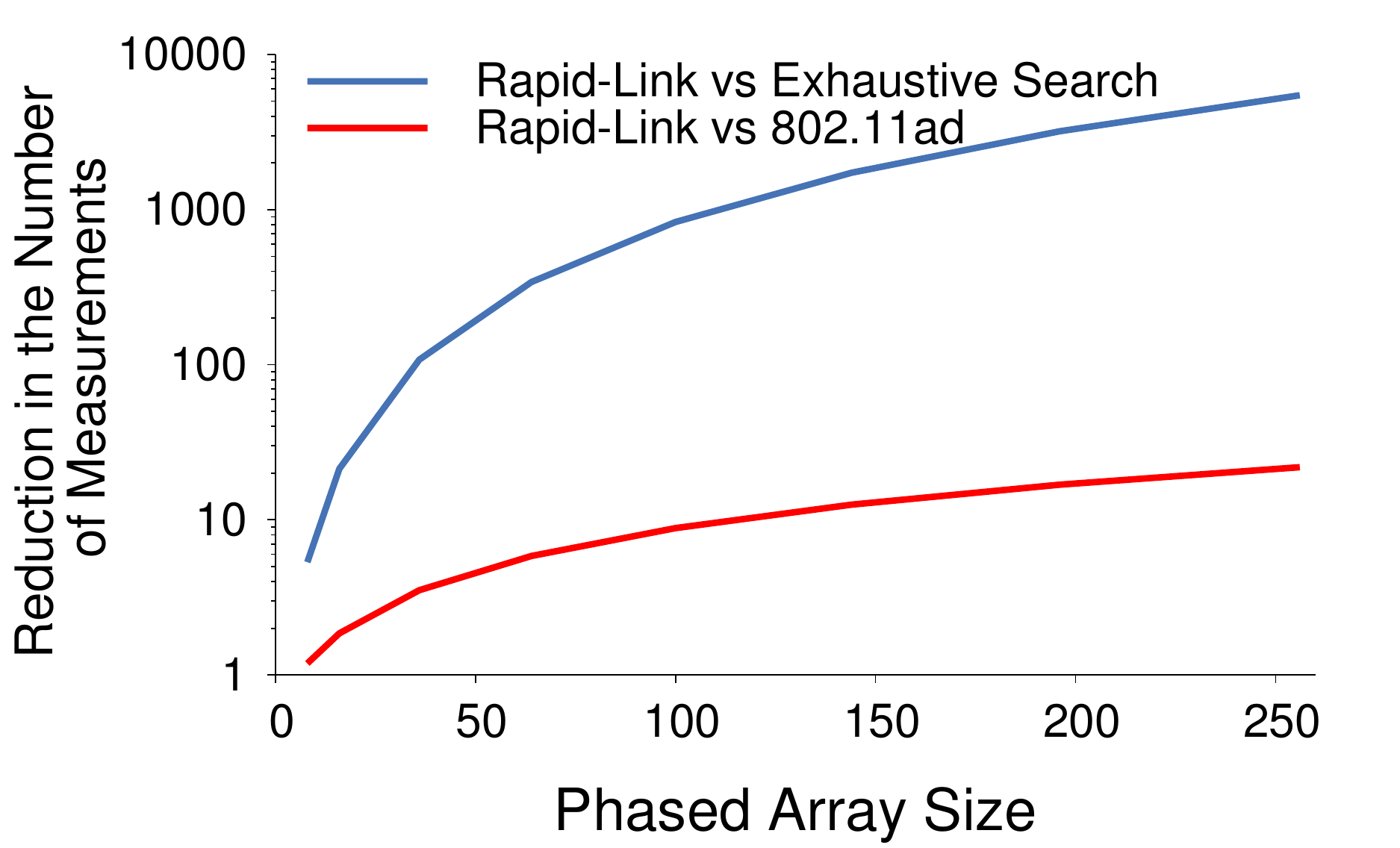}
\vskip -0.1in
\caption{ \textbf{Beam Alignment Latency} \footnotesize The figure
plots the reduction in the number of measurements for \name\ compared to the 802.11ad standard and 
exhaustive search.}
\vskip -0.1in
\label{fig:searchtime_sim}
\end{figure}

\vskip 0.06in\noindent
{\it (b) Reduction in Search Time:} 
Next, we look at the amount of time it takes to find the best alignment in each scheme under the 802.11ad MAC protocol. The standard is still evolving; our description is based on~\cite{802.11ad}.  Since the delays in exhaustive search are unacceptable in practice, we consider only \name\ and the beam alignment scheme in the standard.   

The 802.11ad has a protocol for when the AP and clients search the space to align their beams~\cite{802.11ad,802.11ad_paper}. The delay produced through this process differs from simply multiplying the number of measurement frames by the duration of each measurement. This is due to three  main reasons: 1) The protocol allows for beam scan (called beam training) only during certain intervals. If the client cannot collect all necessary measurements, it needs to wait until the next opportunity to perform more measurements. 2) Different clients contend for the beam alignment slots; hence, the delay will increase depending on the number of clients. 3) When the AP sweeps its beam, all clients can collect measurements; hence this part can be amortized. 

To better understand the above constraints, let us describe at a high-level how 802.11ad performs beam-forming training. The AP periodically transmit beacon intervals (BI), which typically last for 100~ms~\cite{802.11ad_paper}. Each BI has a beacon header intervals (BHI), followed by a data transmission interval (DTI), as shown in Fig.~\ref{fig:BI}. The search for the best alignment is done during the BHI.  Each BHI consists of one beacon transmission interval (BTI) which is used by the AP to train its antenna beam, and eight association beam-forming training (A-BFT) slots, which are randomly selected by clients to train their beams. Finally, each A-BFT slot consists of up to 16 SSW frames, where each frame is used to perform one measurement and has a duration of $15.8 \mu s$~\cite{802.11ad,802.11ay}. Each BI has a maximum of 8 beam training slots. All clients contend for training in those slots.  If the client cannot finish its training during one A-BFT, it can contend for further slots during the same or following BI. Yet, waiting for the next BI increases the delay by 100ms. 

As explained in~\xref{sec:baselines}, 802.11 performs beam refinement where each of $\gamma$ best directions at the AP and client are compared again. To simplify the computation, we conservatively ignore the 802.11ad beam refinement since it only increases the delay of 802.11ad, and improves the relative gains of \name. Also when simulating 4 client, we assume that the contention succeeded without collision. This too is a conservative assumption since \name\ requires significantly fewer measurement slots and hence, given the same number of slots, the collision probability between clients is smaller in \name. Finally, the AP trains its beam during the BTI, and uses frames similar to those used for the client beam training. The AP does not need to repeat this training per client. 

\begin{figure}[t!]
\centering
	\includegraphics[width=1.6in]{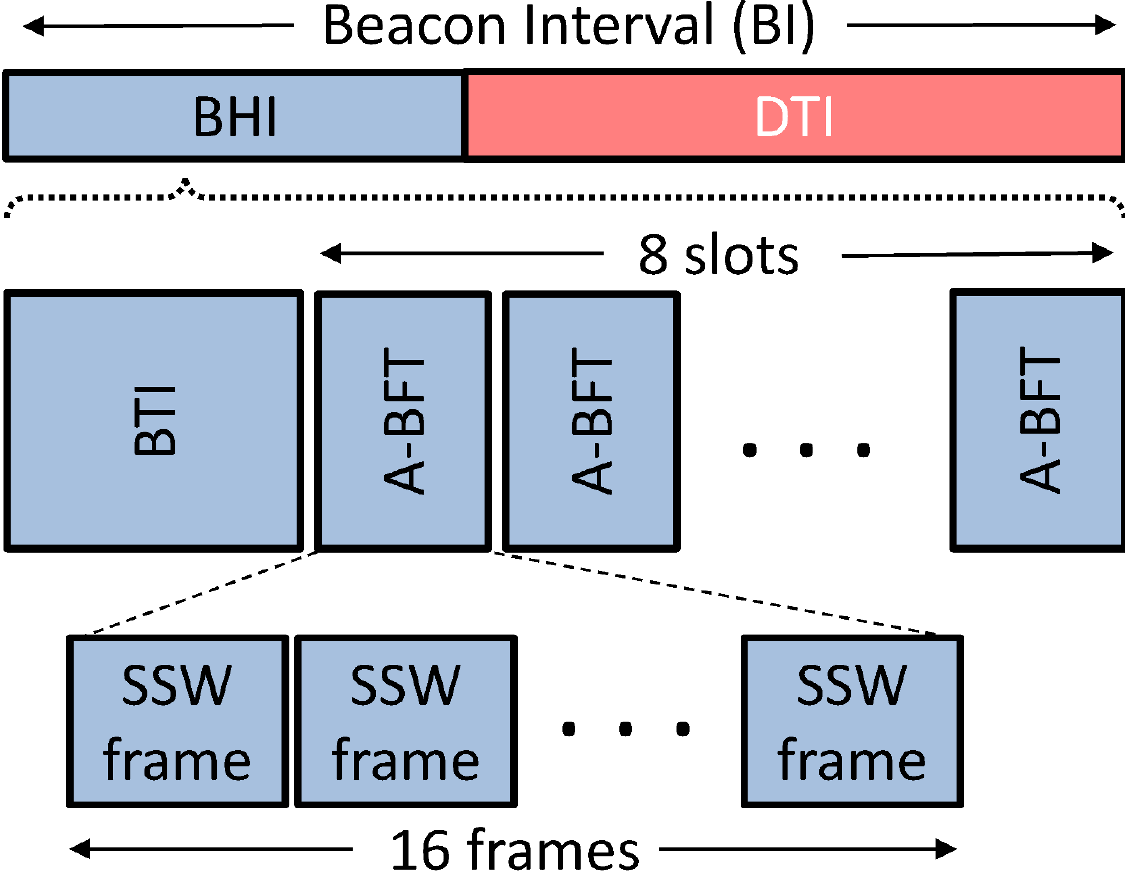}
\vskip -0.15in
\caption{ \footnotesize\textbf{802.11ad Beacon Interval Structure}}
\label{fig:BI}
\end{figure}
 

Table~\ref{tab:latency} shows the beam alignment delay for different antenna-array sizes, both for the case of one client and 4 clients. The table shows that as the number of antennas in the array increases, the delay in 802.11ad increases quickly. In contrast, \name\ can operate within the same standard, but it extracts more information from each measurement, hence keeping the the delay low even for large antenna arrays.   In particular, for antenna arrays of 256 elements, the proposed 802.11 beam alignment algorithm takes hundreds of milliseconds for one user and over 1.5 seconds for 4 users. In contrast, \name\ keeps the delay below 1.01ms and 2.53ms, respectively.

\begin{table}
	\centering
{\small
	\begin{tabular}{|c|c|c|c|c|}
	\hline
	\multirow{3}{*}{Size} & \multicolumn{2}{c|}{One Client} & \multicolumn{2}{c|}{Four Clients} \\
	\cline{2-5} & 802.11ad & \name & 802.11ad & \name \\
	\hline
	 8 & 0.51ms & 0.44ms & 1.27ms & 1.20ms\\
	\hline
	 16 & 1.01ms & 0.51ms & 2.53ms & 1.26ms\\
	\hline
 	 64 & 4.04ms & 0.89ms & 304.04ms & 2.40ms\\
	\hline
	 128 & 106.07ms & 0.95ms & 706.07ms & 2.46ms\\
	\hline
	 256 & 310.11ms & 1.01ms & 1510.11ms & 2.53ms \\
	\hline
\end{tabular}}
\vskip -0.1in
\caption{\footnotesize{\bf Beam alignment latency for different array size}}
\vskip -0.1in
\label{tab:latency}
\end{table}

\section{Conclusion}
This paper presents \name, a phased array mmWave system that can find the best beam alignment without scanning the entire space. \name\ delivers the first mmWave beam alignment algorithm with provably logarithmic measurements for the phased-array architecture commonly used in mmWave access points and clients. Our algorithm finds the correct alignment of the beams between a transmitter and a receiver orders of magnitude faster than existing radios. In particular, it reduces the alignment delay from over a second to 2.5ms. We believe \name\ brings us closer towards practical mmWave networks.

\begin{appendices}
\appendix




\section{Notation and Preliminaries}

\noindent {\bf (a) Basic notation}
\begin{Itemize}
\item We use  $[N]$ to denote $\{0 \ldots N-1\}$.
\item We use $S$ to denote the support of ${\bf x}$.
\item We use ${\bf F}$ to denote the Fourier transform matrix, and ${\bf F'}$ to denote the inverse Fourier transform matrix. Also, we use ${\bf F}_i$ to denote the $i$-th row of ${\bf F}$; same for ${\bf F'}$. Finally, we use $\hat{{\bf x}}$ to denote ${\bf F}{\bf x}$.
\item  For two vectors ${\bf x}$ and ${\bf y}$, we define the {\em Hadamard product} $\circ$ of ${\bf x}$ and ${\bf y}$ as $({\bf x} \circ {\bf y})_i = x_i y_i$. We will use this notion to mask out the coefficients outside of a given segment. 
\item We use ${\bf x * y}$ to denote the convolution of ${\bf x}$ and ${\bf y}$.
\item The vectors ${\bf e}_0 \ldots {\bf e}_{N-1}$ denote the standard basis. I.e., $(e_p)_i=1$ for $p=i$ and $(e_p)_i=0$ otherwise.  
 \end{Itemize}

 \vskip 0.06in\noindent {\bf (b) Measurements and box car filter}\newline
 Our measurements can be described using the notion of the {\em boxcar filter}, defined as follows. For parameter $P$, let ${\bf H}$ be such that $H_i=\frac{\sqrt{N}}{P-1}$ if $| i | \le P/2$ and $H_i=0$ otherwise. It is known that 
 $ \hat{H}_j = \frac{1}{P-1} \frac{sin(\pi(P -1) j/N)}{(P-1) \sin (\pi j/N)} $.
 
\begin{proposition}
\label{p:H}
${\bf \hat{H}}$ satisfies the following properties: (i) $\hat{H}_0 = 1$; (ii)  $\hat{H}_j \in [\frac{1}{2 \pi}, 1]$ for $|j| \le \frac{N}{2P}$; (iii)  $|\hat{H}_j| \le \frac{2}{1+|j| P/N}$ if $P \ge 3$.
\end{proposition} 

\begin{claim}
\label{c:C}
\[ \|{\bf \hat{H}} \|_2^2 = \sum_j |\hat{H}_j|^2 \le  1+2N/P \sum_j 1/|j|^2 \le  C \frac{N}{P}  \] 
for some constant $C$.
\end{claim}

We also define a {\em shifted} version of ${\bf H}$ defined as $(H^t)_i = H_{i-t}$. By the time-shift theorem it follows that $|\hat{H}^t_i| = \hat{H}_i$.

Using this notation, we can write each measurement ${\bf a}^b$ as

\[ {\bf a}^b = \sum_{r=0}^{R-1} ( {\bf F}_{s^r_b} \circ {\bf H}^{r N/R} ) \omega^{t_r^b} \]

Each segment of ${\bf a^b}$, when multiplied by a row of the matrix ${\bf F'}$, can be interpreted as follows. 

\begin{claim}
\label{c:measurement}
$({\bf F}_i \circ {\bf H}) \cdot {\bf F'}_p = \hat{H}_{i-p}$
\end{claim}

\begin{proof}
\[ ({\bf F}_i \circ {\bf H}) \cdot {\bf F'}_p
 = {\bf F}_i \cdot ( {\bf F'}_p \circ {\bf H}) 
 = \widehat{({\bf F'}_p \circ {\bf H})}_i 
 = (\hat{{\bf F'}_p} * \hat{{\bf H}})_i 
 = ({\bf e}_p *  \hat{{\bf H}})_i 
 = \hat{H}_{i-p}
\]
\end{proof}

 \vskip 0.06in\noindent {\bf (c) Pseudo-random permutations} \newline
 We will use matrices ${\bf P}_{\rho,b}$ parameterized by mappings $\rho$ of the form      $\rho(i)=\sigma^{-1} i +a \mod N$ for $\sigma, a, b \in [N]$ such that 
\begin{Itemize}
\item $({\bf P}_{\rho,b}{\bf x})_{\rho(i)} = x_{i} \omega^{\tau(j)}$ for $\tau(j)=b(j+\sigma a)$
\item  ${\bf F' P}_{\rho,b}  = {\bf P'}_{\rho,b} {\bf F}$ for ${\bf P'}_{\rho,b}$ as defined in Section~\ref{sec:alg}. 
\end{Itemize}
Note that $\tau$ is a permutation assuming $b$ is invertible mod $N$. 
We use ${\cal R}$ to denote the set of all mappings $\rho$ as defined above. 
For the analysis, we will assume that $N$ is prime. This will ensure that the elements $\rho \in {\cal R}$ are permutations. Furthermore, in this case ${\cal R}$ is {\em pairwise independent}, i.e., for any $i \neq j, i' \neq  j'$, we have 
\[Pr_{\rho \in {\cal R}} [\rho(i)=i', \rho(j)=j'] = 1/N^2 \]

It will be convenient to assume $\|{\bf x}\|^2_2=1$.  Then we can define the threshold $T$ to be $ ( \frac{1}{2(2 \pi)} -  \frac{1}{8 \pi})^2 (\frac{1}{2(2 \pi)})^2 /K   $.

\section{Proofs}

\begin{lemma}
\label{c:tail}
Fix $b$ and select $\rho \in {\cal R}$ uniformly at random. Then, for any $s$:
\[  E[ I(b,\rho(s))] = E\left[ |{\bf a}^b {\bf F'}_{\rho(s)}|^2 \right ] \le  C R/P  \]
where $C$ is the constant from Claim~\ref{c:C}.
\end{lemma}

\begin{proof}
\begin{eqnarray}
E\left[ |{\bf a}^b {\bf F'}_{\rho(s)}|^2 \right ] & = &  E\left[ |\sum_{r=0}^{R-1} (({\bf F}_{s^r_b} \circ {\bf H}^{r N/R}) \omega^{t_r^b} ) {\bf F'}_{\rho(s)}|^2\right] \\
 \label{e:step}
& = &  E\left[ |\sum_{r=0}^{R-1}  \hat{H}_{s^r_b - \rho(s)}^{r N/R}  \omega^{t_r^b}|^2 \right]\\
& = &    \sum_{r=0}^{R-1} E\left[|\hat{H}_{s^r_b - \rho(s)}|^2 \right]
\end{eqnarray}
where in step \ref{e:step} we used Claim~\ref{c:measurement} and the independence of the variables $t_r^b$, $r=0 \ldots R-1$. Since $i=s^r_b - \rho(s)$ is distributed uniformly at random in  $[N]$, by Claim~\ref{c:C}:
 
 \[  \sum_{r=0}^{R-1}  E\left[|\hat{H}_{s^r_b - \rho(s)}|^2 \right]  \le  \sum_{r=0}^{R-1} 1/N \sum_i |\hat{H}_{i} |^2
  \le R/N  \cdot N/P  \cdot C= C R/P 
  \]
\end{proof}

\begin{lemma} Suppose that $|s^r_b - i| \le \frac{N}{2P}$. Then 
\[  \Pr [ |{\bf a}^b {\bf F'}_i|^2  \ge \frac{1}{4(2 \pi)^2}] \ge 5/6\]
\label{l:non}
\end{lemma}

\begin{proof}

\begin{eqnarray}
|{\bf a}^b {\bf F'}_i|^2 & = & |\sum_{r'=0}^{R-1}  \hat{H}_{s^{r'}_b - i}^{r' N/R} \omega^{t_{r'}^b}|^2  \\
& = & | \hat{H}_{s^r_b-i} \omega^{t_r^b} +  \sum_{r'\neq r} \hat{H}_{s^{r'}_b - i} \omega^{t_{r'}^b}|^2  \\
& = & | \hat{H}_{s^r_b-i} \omega^{t_r^b} -  X|^2
\end{eqnarray}
We know from Proposition~\ref{p:H} that $| \hat{H}_{s^r_b-i} \omega^{t_r^b}| \ge \frac{1}{2 \pi}$. We will show that the probability of $X \ge  \frac{1}{2(2 \pi)}$ is at most $1/6$. 
It will follow that $|{\bf a}^b {\bf F'}_i|^2 \ge \frac{1}{4(2 \pi)^2}$ with the probability of at least $5/6$.

Recall that $s^0_b, s^1_b \ldots $ are separated by $P$. Therefore, for $r' \neq r$, we have $|s^{r'}_b - i| \ge P - |s^r_b - i| \ge P- \frac{N}{2P}$, which is at least $P/2$ for $B$ large enough. By the independence of the variables $t_{r'}^b$ and by Proposition~\ref{p:H} we have:
\begin{eqnarray*}
E[ X^2] & = & E[ |\sum_{r'\neq r} \hat{H}_{s^{r'}_b - i}  \omega^{t_{r'}^b} |^2 ]\\
&  =  & \sum_{r'\neq r} E[ |\hat{H}_{s^{r'}_b - i}  \omega^{t_{r'}^b} |^2 ]\\
& \le & 2 \sum_{d=1}^R \left ( \frac{2}{1+P/N \cdot d(P/2)} \right )^2\\
& \le & 8/(P^2 /N)   \sum_{d=1}^R (1/d)^2 \le 8 C N/P^2
\end{eqnarray*}

Since $P=N/R$ and $R^2=N/B$, the latter expression is bounded by $8 C/B$, which is less than 
$\frac{1}{6(4\pi)^2}$ for $B$ large enough. It follows that the probability of $X^2 \ge \frac{1}{(4\pi)^2} = 6 E[X^2]$ is at most $1/6$.
\end{proof}

\begin{proof}[of Theorem~\ref{t:main}, Part (1)] Suppose that $s \in S$. Select $s^r_b$ that is closest to $\rho(s)$. Note that $|s^r_b - \rho(s)| \le \frac{N}{2P}$, which by Lemma~\ref{l:non} implies $|{\bf a}^b {\bf  F'}_{\rho(s)} |^2 \ge \frac{1}{ 4(2 \pi)^2}$ with the probability of at least $5/6$.

We now lower bound  $T(s)$ as follows

\begin{eqnarray}
T(s) &  \ge & |{\bf a}^b {\bf F'} {\bf P}_{\rho,b} {\bf x}|^2 |{\bf a}^b {\bf F'} {\bf P}_{\rho,b} {\bf e}_s|^2  \\
& = & |Yx_s -X|^2 |Y|^2
\end{eqnarray}
where $Y= \omega^{\tau(s) } {\bf a}^b {\bf F'}_{\rho(s)}$ and \\
$X=\sum_{s' \in S-\{s\}} \omega^{\tau(s') } {\bf a}^b {\bf F'}_{\rho(s')} x_{s'} $.

We can bound $E[|X|^2]$ as follows.
\begin{eqnarray}
& & E[ |\sum_{s' \in S-\{s\}}  \omega^{\tau(s')} x_{s'} {\bf a}^b {\bf F'}_{\rho(s')}|^2  ]\\
& =   & \sum_{s' \in S-\{s\}} x^2_{s'}  E[|{\bf a}^b {\bf F'}_{\rho(s')}|^2] \\
& \le & \sum_{s' \in S-\{s\}} x^2_{s'}  C R/P \\
& = &  \|{\bf x}\|_2^2 C R/P  \le C/B \le 1/K \frac{1}{6\cdot  (8 \pi^2)^2}
\end{eqnarray}
where we used Parseval's identity, Lemma \ref{c:tail} and that $B$ is large enough.
Therefore, we have that $\Pr[X^2 \ge  \frac{1}{ (8 \pi)^2}] \le 1/6$. By Lemma~\ref{l:non} we have that, with probability at least $1-1/6-1/6$,  $ T(s) \ge ( \frac{1}{2(2 \pi)} -  \frac{1}{8 \pi})^2 (\frac{1}{2(2 \pi)})^2/K $.
\end{proof}

\begin{proof}[of Theorem~\ref{t:main}, Part (2)]  Suppose that $s \notin S$. We have
\begin{eqnarray*}
& & E[T(s)]\\
 & \le &  \sum_{b=0}^{B-1} E_{\rho(s), \rho(s'), \tau} [ |\sum_{s' \in S} \omega^{\tau(s')} x^2_{s'}  {\bf a}^b {\bf F'}_{\rho(s')}|^2    |{\bf a}^b {\bf F'}_{\rho(s)}|^2]  \\
& = & \sum_{b=0}^{B-1}  E_{\rho(s'), \tau} [ |\sum_{s' \in S} \omega^{\tau(s')} x^2_{s'}  {\bf a}^b {\bf F'}_{\rho(s')}|^2] E_{\rho(s)} [ |{\bf a}^b {\bf F'}_{\rho(s)}|^2]\\
& \le & CR/P \sum_{b=0}^{B-1} \sum_{s' \in S} x^2_{s'} E [ |{\bf a}^b {\bf F'}_{\rho(s')}|^2]\\
& \le & (CR/P)^2 B \|{\bf x}\|_2^2 \le C^2 /B \le T/3
\end{eqnarray*}
where we assume that $B$ is large enough. 
By Markov inequality it follows that $\Pr[T(s) \ge T ] \le 1/3$.
\end{proof}

\end{appendices}


\let\oldthebibliography=\thebibliography
\let\endoldthebibliography=\endthebibliography
\renewenvironment{thebibliography}[1]{%
    \begin{oldthebibliography}{#1}%
      \setlength{\parskip}{0ex}%
      \setlength{\itemsep}{0ex}%
}%
{%
\end{oldthebibliography}%
}
{
\bibliographystyle{abbrv}
\bibliography{ourbib}

\begin{thebibliography}{10}

\bibitem{FCC_24ghz}
{\em FCC to explore 5G services above 24 GHz}.
\newblock
  http://www.fiercewireless.com/tech/fcc-to-explore-5g-services-auctioned-or-unlicensed-above-24-ghz.

\bibitem{infineon_24ghz}
{\em mmWave 24-86 GHz Transceivers}.
\newblock www.infineon.com.

\bibitem{Wilocity}
{Wilocity 802.11ad Multi-Gigabit Wireless Chipset}.
\newblock http://wilocity.com, 2013.

\bibitem{802.11ay}
{Short SSW Format for 11ay}.
\newblock
  https://mentor.ieee.org/802.11/dcn/16/11-16-0416-01-00ay-short-ssw-format-for-11ay.pptx,
  2016.

\bibitem{siBeam_VR}
{60 GHz: Taking the VR Experience to the Next Level}.
\newblock http://www.sibeam.com.

\bibitem{MoVR_hotnets}
O.~Abari, D.~Bharadia, A.~Duffield, and D.~Katabi.
\newblock Cutting the cord in virtual reality.
\newblock In {\em HotNets}, 2016.

\bibitem{mmw_hybrid1}
A.~Alkhateeb, O.~El~Ayach, G.~Leus, and R.~W. Heath.
\newblock {Channel estimation and hybrid precoding for millimeter wave cellular
  systems}.
\newblock {\em Selected Topics in Signal Processing, IEEE Journal of}, 2014.

\bibitem{mmw_measurement2}
C.~R. Anderson and T.~S. Rappaport.
\newblock {In-Building Wideband Partition Loss Measurements at 2.5 and 60 GHz}.
\newblock {\em IEEE Transactions on Wireless Communications}, 3(3), May 2004.

\bibitem{mmw_heir2}
D.~C. Ara{\'u}jo, A.~L. de~Almeida, J.~Axnas, and J.~Mota.
\newblock Channel estimation for millimeter-wave very-large mimo systems.
\newblock In {\em Signal Processing Conference (EUSIPCO), 2014 Proceedings of
  the 22nd European}, pages 81--85. IEEE, 2014.

\bibitem{ba2010lower}
K.~D. Ba, P.~Indyk, E.~Price, and D.~P. Woodruff.
\newblock Lower bounds for sparse recovery.
\newblock In {\em Proceedings of the twenty-first annual ACM-SIAM symposium on
  Discrete Algorithms}, pages 1190--1197. SIAM, 2010.

\bibitem{IBM_array}
{BM Breakthrough Could Alleviate Mobile Data Bottleneck}.
\newblock
  http://cacm.acm.org/news/164893-ibm-breakthrough-could-alleviate-mobile-data-bottleneck/fulltext.

\bibitem{qualcomm}
M.~Branda.
\newblock {Qualcomm Research demonstrates robust mmWave design for 5G}.
\newblock Qualcomm Technologies Inc., November 2015.

\bibitem{CRT}
E.~Candes, J.~Romberg, and T.~Tao.
\newblock Robust uncertainty principles: Exact signal reconstruction from
  highly incomplete frequency information.
\newblock {\em IEEE Transactions on Information Theory}, 52:489 -- 509, 2006.

\bibitem{cisco}
{Cisco}.
\newblock {Cisco Visual Networking Index: Global Mobile Data Traffic Forecast
  Update}, 2013.

\bibitem{mmw_DC3}
Y.~Cui, S.~Xiao, X.~Wang, Z.~Yang, C.~Zhu, X.~Li, L.~Yang, and N.~Ge.
\newblock {Diamond: Nesting the Data Center Network with Wireless Rings in 3D
  Space}.
\newblock In {\em NSDI}, 2016.

\bibitem{5Gvision}
L.~DMC R\&D~Center, Samsung Electronics~Co.
\newblock {5G Vision White Paper - Samsung}.
\newblock 2015.

\bibitem{Don}
D.~Donoho.
\newblock Compressed sensing.
\newblock {\em IEEE Transactions on Information Theory}, 52(4):1289 -- 1306,
  2006.

\bibitem{mmw_hybrid2}
M.~E. Eltayeb, A.~Alkhateeb, R.~W. Heath, and T.~Y. Al-Naffouri.
\newblock {Opportunistic Beam Training with Hybrid Analog/Digital Codebooks for
  mmWave Systems}.
\newblock In {\em GLOBESIP}, 2015.

\bibitem{ericsson}
{Ericsson}.
\newblock Traffic and market data report, 2011.

\bibitem{mmw_cs2}
B.~Gao, Z.~Xiao, C.~Zhang, D.~Jin, and L.~Zeng.
\newblock {Sparse/dense channel estimation with non-zero tap detection for
  60-GHz beam training}.
\newblock {\em IET Communications}, 8(11):2044--2053, 2014.

\bibitem{GSM}
A.~Gilbert, M.~Muthukrishnan, and M.~Strauss.
\newblock Improved time bounds for near-optimal space fourier representations.
\newblock In {\em SPIE}, 2005.

\bibitem{gilbert2002near}
A.~C. Gilbert, S.~Guha, P.~Indyk, S.~Muthukrishnan, and M.~Strauss.
\newblock Near-optimal sparse fourier representations via sampling.
\newblock In {\em Proceedings of the thiry-fourth annual ACM symposium on
  Theory of computing}, pages 152--161. ACM, 2002.

\bibitem{gilbert2014recent}
A.~C. Gilbert, P.~Indyk, M.~Iwen, and L.~Schmidt.
\newblock Recent developments in the sparse fourier transform: a compressed
  fourier transform for big data.
\newblock {\em IEEE Signal Processing Magazine}, 31(5):91--100, 2014.

\bibitem{MOCA}
M.~K. Haider and E.~W. Knightly.
\newblock {Mobility resilience and overhead constrained adaptation in
  directional 60 GHz WLANs: protocol design and system implementation}.
\newblock In {\em MobiHoc}, 2016.

\bibitem{flyways}
D.~Halperin, S.~Kandula, J.~Padhye, P.~Bahl, and D.~Wetherall.
\newblock {Augmenting Data Center Networks with Multi-Gigabit Wireless Links}.
\newblock In {\em ACM SIGCOMM}, 2011.

\bibitem{mmW5G}
S.~Han, C.~I, Z.~Xu, and C.~Rowell.
\newblock {Large-Scale Antenna Systems with Hybrid Analog and Digital
  Beamforming for Millimeter Wave 5G}.
\newblock {\em IEEE Communications Magazine}, January 2015.

\bibitem{SFFT2}
H.~Hassanieh, P.~Indyk, D.~Katabi, and E.~Price.
\newblock Nearly optimal sparse fourier transform.
\newblock In {\em STOC}, 2012.

\bibitem{SFFT1}
H.~Hassanieh, P.~Indyk, D.~Katabi, and E.~Price.
\newblock Simple and practical algorithm for sparse \uppercase{FFT}.
\newblock In {\em SODA}, 2012.

\bibitem{hosoya2015multiple}
K.~Hosoya, N.~Prasad, K.~Ramachandran, N.~Orihashi, S.~Kishimoto,
  S.~Rangarajan, and K.~Maruhashi.
\newblock Multiple sector id capture (midc): A novel beamforming technique for
  60-ghz band multi-gbps wlan/pan systems.
\newblock {\em IEEE Transactions on Antennas and Propagation}, 2015.

\bibitem{802.11ad}
{IEEE Standards Association}.
\newblock {IEEE Standards 802.11ad-2012: Enhancements for Very High Throughput
  in the 60 GHz Band}, 2012.

\bibitem{iwen2017robust}
M.~Iwen, A.~Viswanathan, and Y.~Wang.
\newblock Robust sparse phase retrieval made easy.
\newblock {\em Applied and Computational Harmonic Analysis}, 42(1):135--142,
  2017.

\bibitem{jaganathan2015phase}
K.~Jaganathan, Y.~C. Eldar, and B.~Hassibi.
\newblock Phase retrieval: An overview of recent developments.
\newblock {\em arXiv preprint arXiv:1510.07713}, 2015.

\bibitem{mmw_AD}
J.~Kilpatrick, R.~Shergill, and M.~Sinha.
\newblock {60 GHz Line of Sight Backhaul Links Ready to Boost Cellular
  Capacity}.
\newblock Analog Devices Inc.

\bibitem{mmw_heir4}
J.~Kim and A.~F. Molisch.
\newblock {Fast Millimeter-Wave Beam Training with Receive Beamforming}.
\newblock {\em Journal of Communications and Networks}, 16(5), October 2014.

\bibitem{mmw_heir1}
B.~Li, Z.~Zhou, W.~Zou, X.~Sun, and G.~Du.
\newblock {On the Efficient Beam-Forming Training for 60GHz Wireless Personal
  Area Networks}.
\newblock {\em IEEE Transactions on Wireless Communications}, 12(2), February
  2013.

\bibitem{BoonBane}
T.~Nitsche, G.~Bielsa, I.~Tejado, A.~Loch, and J.~Widmer.
\newblock Boon and bane of 60 ghz networks: Practical insights into
  beamforming, interference, and frame level operation.
\newblock In {\em Proceedings of the 11th ACM Conference on Emerging Networking
  Experiments and Technologies}. ACM, 2015.

\bibitem{802.11ad_paper}
T.~Nitsche, C.~Cordeiro, A.~B. Flores, E.~W. Knightly, E.~Perahia, and J.~C.
  Widmer.
\newblock {IEEE 802.11 ad: directional 60 GHz communication for
  multi-Gigabit-per-second Wi-Fi}.
\newblock {\em IEEE Communications Magazine}, 2014.

\bibitem{BBS}
T.~Nitsche, A.~B. Flores, E.~W. Knightly, and J.~Widmer.
\newblock Steering with eyes closed: mm-wave beam steering without in-band
  measurement.
\newblock In {\em INFOCOM}, 2015.

\bibitem{pi2011introduction}
Z.~Pi and F.~Khan.
\newblock An introduction to millimeter-wave mobile broadband systems.
\newblock {\em Communications Magazine, IEEE}, 2011.

\bibitem{price2011}
E.~Price and D.~P. Woodruff.
\newblock (1+ eps)-approximate sparse recovery.
\newblock In {\em Foundations of Computer Science (FOCS), 2011 IEEE 52nd Annual
  Symposium on}, pages 295--304. IEEE, 2011.

\bibitem{MegaMIMO}
H.~Rahul, S.~S. Kumar, and D.~Katabi.
\newblock Megamimo: Scaling wireless capacity with user demand.
\newblock In {\em SIGCOMM}, 2012.

\bibitem{mmw_cs1}
D.~Ramasamy, S.~Venkateswaran, and U.~Madhow.
\newblock Compressive tracking with 1000-element arrays: A framework for
  multi-gbps mm wave cellular downlinks.
\newblock In {\em Communication, Control, and Computing (Allerton), 2012 50th
  Annual Allerton Conference on}. IEEE, 2012.

\bibitem{rangan2014millimeter}
S.~Rangan, T.~S. Rappaport, and E.~Erkip.
\newblock Millimeter-wave cellular wireless networks: Potentials and
  challenges.
\newblock {\em Proceedings of the IEEE}, 2014.

\bibitem{samsung}
W.~Roh, J.-Y. Seol, J.~Park, B.~Lee, J.~Lee, Y.~Kim, J.~Cho, K.~Cheun, and
  F.~Aryanfar.
\newblock {Millimeter-Wave Beamforming as an Enabling Technology for 5G
  Cellular Communications: Theoretical Feasibility and Prototype Results}.
\newblock {\em IEEE Communications Magazine}, February 2014.

\bibitem{argos}
C.~Shepard, H.~Yu, N.~Anand, L.~E. Li, T.~Marzetta, R.~Yang, and L.~Zhong.
\newblock {Argos: Practical Many-Antenna Base Stations}.
\newblock In {\em MobiCom}, 2012.

\bibitem{argos2}
C.~Shepard, H.~Yu, and L.~Zhong.
\newblock {ArgosV2: A Flexible Many-Antenna Research Platform}.
\newblock In {\em MobiCom}, 2013.

\bibitem{SiBeam}
{SiBeam, Lattice Semiconductor}.
\newblock www.sibeam.com.

\bibitem{WiMi}
S.~Sur, V.~Venkateswaran, X.~Zhang, and P.~Ramanathan.
\newblock {60 GHz Indoor Networking through Flexible Beams: A Link-Level
  Profiling}.
\newblock In {\em SIGMETRICS}, 2015.

\bibitem{BeamSpy}
S.~Sur, X.~Zhang, P.~Ramanathan, and R.~Chandra.
\newblock {BeamSpy: Enabling Robust 60 GHz Links Under Blockage}.
\newblock In {\em NSDI}, 2016.

\bibitem{TPLink}
{TPLink}.
\newblock {Talon AD7200 Wireless Wi-Fi Tri-Band Gigabit Router}.
\newblock http://www.tp-link.com/en/products/details/AD7200.html.

\bibitem{mmw_heir7}
Y.~M. Tsang, A.~S.~Y. Poon, and S.~Addepalli.
\newblock {Coding the Beams: Improving Beamforming Training in mmWave
  Communication System}.
\newblock In {\em IEEE GLOBECOM}, 2011.

\bibitem{TseV:05}
D.~Tse and P.~Vishwanath.
\newblock {\em Fundamentals of Wireless Communications}.
\newblock Cambridge University Press, 2005.

\bibitem{umts}
{UMTS Forum}.
\newblock Mobile traffic forecasts: 2010-2020 report, 2011.

\bibitem{mmw_heir5}
J.~Wang, Z.~Lan, C.-W. Pyo, T.~Baykas, C.-S. Sum, M.~A. Rahman, J.~Gao,
  R.~Funada, F.~Kojima, H.~Harada, and S.~Kato.
\newblock {Beam Codebook Based Beamforming Protocol for Multi-Gbps
  Millimeter-Wave WPAN Systems}.
\newblock {\em IEEE Journal of Selected Areas in Communications}, 27(8),
  October 2009.

\bibitem{ArrayTrack}
J.~Xiong and K.~Jamieson.
\newblock Array{T}rack: A fine-grained indoor location system.
\newblock In {\em NSDI}, 2013.

\bibitem{mmw_heir3}
W.~Yuan, S.~M.~D. Armour, and A.~Doufexi.
\newblock {An Efficient and Low-complexity Beam Training Technique for mmWave
  Communication}.
\newblock In {\em PIMRC}, 2015.

\bibitem{OpenMili}
J.~Zhang, X.~Zhang, P.~Kulkarni, and P.~Ramanathan.
\newblock {OpenMili: A 60 GHz Software Radio Platform With a Reconfigurable
  Phased-Array Antenna}.
\newblock In {\em MobiCom}, 2016.

\bibitem{mmw_heir6}
L.~Zhou and Y.~Ohashi.
\newblock {Efficient Codebook-Based MIMO Beamforming for Millimeter-Wave
  WLANs}.
\newblock In {\em PIMRC}, 2012.

\bibitem{mmw_DC2}
X.~Zhou, Z.~Zhang, Y.~Zhu, Y.~Li, S.~Kumar, A.~Vahdat, B.~Y. Zhao, and
  H.~Zheng.
\newblock {Mirror Mirror on the Ceiling: Flexible Wireless Links for Data
  Centers}.
\newblock In {\em ACM SIGCOMM}, 2012.

\bibitem{60GHzMobicom}
Y.~Zhu, Z.~Zhang, Z.~Marzi, C.~Nelson, U.~Madhow, B.~Y. Zhao, and H.~Zheng.
\newblock {Demystifying 60GHz Outdoor Picocells}.
\newblock In {\em MOBICOM}, 2014.

\end{thebibliography}
}

\end{sloppypar}
\end{document}